\newif\ifEditMode
\EditModefalse

\documentclass[runningheads]{llncs}
\usepackage{graphicx, amsmath, amsfonts, tikz, mathtools, tikz-cd,epstopdf,verbatim, hyperref}
\usepackage{wrapfig,xspace,subcaption}
\usepackage[font=itshape]{quoting}
\usepackage{enumitem}
\usepackage{amsbsy,scalerel}
\usepackage[T1]{fontenc}
\usetikzlibrary{automata, positioning, arrows}

\tikzset{
->, 
node distance=3cm, 
every state/.style={thick, fill=gray!10}, 
initial text=$ $, 
}

\PassOptionsToPackage{bookmarks=false}{hyperref}

\graphicspath{{figures/}}

\allowdisplaybreaks

\ifEditMode
   \paperwidth=\dimexpr \paperwidth + 8cm\relax
   \oddsidemargin=\dimexpr\oddsidemargin + 4cm\relax
   \evensidemargin=\dimexpr\evensidemargin + 4cm\relax
   \marginparwidth=\dimexpr \marginparwidth + 4cm\relax
\fi

\usepackage{xcolor}

\newcommand{\myparagraph}[1]{\noindent\emph{#1.}}

\newcommand{\setunion}[0]{\boldsymbol{\cup}}
\newcommand{\setint}[0]{\boldsymbol{\cap}}
\newcommand{\meet}[0]{\wedge}
\newcommand{\bigmeet}[0]{\bigwedge}
\newcommand{\join}[0]{\vee}
\newcommand{\bigjoin}[0]{\bigvee}
\newcommand{\setArg}[2]{\left\{ {#1} \,\,\middle|\,\, {#2} \right\}}

\newcommand{\cmp}[0]{M}
\newcommand{\cmpcompositionId}[0]{\mathbf{Id}}
\newcommand{\cmpcompositionZero}[0]{\mathbf{0}}
\newcommand{\cmpEnv}[0]{E}
\newcommand{\cmpUniverse}[0]{\mathbb{M}}
\newcommand{\behaviorUnv}[0]{\mathcal{B}}
\newcommand{\cmpComposition}[0]{\parallel}

\newcommand{\cmpSet}[0]{H}
\newcommand{\compset}[0]{compset\xspace}
\newcommand{\compsets}[0]{compsets\xspace}

\newcommand{\Compsets}[0]{Compsets\xspace}
\newcommand{\envSet}[0]{\mathcal{E}}
\newcommand{\impSet}[0]{\mathcal{I}}
\newcommand{\sysSet}[0]{\mathcal{S}}

\newcommand{\satImpl}[0]{\models^{I}}
\newcommand{\satEnv}[0]{\models^{E}}
\newcommand{\cmpQuotient}[0]{/}

\newcommand{\intAut}[0]{A}
\newcommand{\maxcomp}[1]{\overline{#1}}
\newcommand{\transition}[2]{\stackrel{#1}{\to_{#2}}}
\newcommand{\defined}[0]{\stackrel{\text{def}}{=}}

\newcommand{\inSpace}[0]{\mathcal{X}}
\newcommand{\outSpace}[0]{\mathcal{Y}}

\newcommand{\cont}{\mathcal{C}}

\newcommand{\contComposition}[0]{\parallel}

\newcommand{\contQuotient}[0]{/}

\newcommand{\recip}[1]{{#1}^{-1}}

\newcommand{\catCmpSet}[0]{\mathbf{CmpSet}}

\newcommand{\catContract}[0]{\mathbf{Contr}}

\renewcommand{\paragraph}[1]{\noindent\textbf{#1}}

\renewcommand{\tt}{\texttt}

\newcommand{\alphabet}{\Sigma}
\newcommand{\langSet}{\mathcal{L}}
\newcommand{\estring}{\epsilon}
\newcommand{\concat}[0]{\circ}
\newcommand{\pre}[0]{\text{Pre}}

\newcommand{\missingExt}[0]{\text{MissExt}}
\newcommand{\uncont}[0]{\text{Unc}}

\newcommand{\intAutClass}[0]{\mathcal{A}}
\newcommand{\lang}[1]{\ell\left(#1\right)}
\newcommand{\iaEnv}[1]{E_{\lang{#1}}}
\newcommand{\iaImp}[1]{M_{\lang{#1}}}

\newcommand{\maximal}{conic\xspace}
\newcommand{\conic}{conic\xspace}
\newcommand{\Conic}{Conic\xspace}

\begin{document}
\title{Hypercontracts}






\author{Inigo Incer\inst{1}
\and
Albert Benveniste\inst{2}
\and
Alberto Sangiovanni-Vincentelli\inst{1}
\and 
Sanjit Seshia\inst{1}
}
\authorrunning{I. Incer et al.}
%
\institute{University of California, Berkeley, USA
\and
INRIA/IRISA, Rennes, France
}

%


%

\maketitle
\begin{abstract}
Contract 
theories have been proposed to formally support
distributed and decentralized system design while ensuring safe system
integration. 
In this paper we propose \emph{hypercontracts}, a generic model with
a richer structure for its underlying model of components, subsuming simulation
preorders. While this new model remains generic, it provides a much more elegant
and richer algebra for its key notions of refinement, parallel composition, and
quotient, and it allows inclusion of new operations. On top of these foundations,
we propose \emph{\conic hypercontracts}, which are still generic but come with a
finite description. 
\end{abstract}

\section{Introduction}
The need for compositional algebraic frameworks to design and analyze
Cyber-Physical Systems is widely accepted. The aim is to support distributed
and decentralized system design based on a proper definition of
\emph{interfaces} supporting the specification of subsystems having a partially
specified context of operation, and subsequently guaranteeing safe system
integration. Over the last few decades, we have seen the introduction of several
formalisms to do this: interface automata \cite{alfaroHenzingerIntAutomata,AlfaroH01,DoyenHJP08,LuttgenV13_LMCS,BujtorV14_Sofsem},
process spaces \cite{NegulescuProcessSpaces}, modal interfaces
\cite{LarsenNW06,LarsenNW07CONCUR,LarsenNW07ESOP,modalInterfaces,DBLP:journals/scp/BauerLLNW14}, assume-guarantee (AG) contracts \cite{Benveniste2008},
rely-guarantee reasoning~\cite{Jones83,Jones2003,ColemanJ07,DBLP:conf/setss/HayesJ17}, and variants of these.
The interface specifications state $(i)$ what the component guarantees and $(ii)$ what
it assumes from its environment in order for those guarantees to hold, i.e.,
all these frameworks implement a form of assume-guarantee reasoning.

These algebraic frameworks have a notion of a component, of an environment, and
of a specification, also called contract to stress the give-and-take dynamics
between the component and its environment. They all have notions of satisfaction
of a specification by a component, and of contract composition. Among the
various contract theories, assume-guarantee contracts require users to
 state the assumptions and guarantees of the specification explicitly, while
interface theories express a specification as a game played between the
specification environments and implementations. Experience tells that
engineers in industry find  the explicit expression of a contract's
assumptions and guarantees natural (see \cite{BenvenisteContractBook} chapter 12), while
interface theories are perceived as a less intuitive mechanism for writing
specifications; however, interface theories in general come with the most
efficient algorithms, making them excellent candidates for internal
representations of specifications. Some authors (\cite{BenvenisteContractBook}
chapter 10) have therefore proposed to translate contracts expressed as pairs
(assumptions, guarantees) into some interface model, where algorithms are
applied. This approach has the drawback that results cannot be traced back to the
original  (assumptions, guarantees) formulation.

The most basic definition of a property in the
formal methods community is ``a set of traces.'' This notion is based on the
behavioral approach to system modelling: we assume we start with a set of
behaviors $\behaviorUnv$, and properties are defined as subsets of $\behaviorUnv$.
In this approach, design elements or components are also defined as subsets of
$\behaviorUnv$. The difference between components and properties is semantics: a
component collects the behaviors that can be observed from that component, while
a property collects the behaviors meeting some criterion of interest. We say a
component $M$ satisfies a property $P$, written $M \models P$, when $M \subseteq
P$, that is, when the behaviors of $M$ meet the criterion that determines $P$.
Properties of this sort are also called \emph{trace properties}. Many design
qualities are of this type, such as \emph{safety}. But there are many system attributes
that can only be determined by analyzing multiple traces such as mean response times,
security attributes, and reliability. This suggests the need for a richer
formalism for expressing design attributes: hyperproperties.

\emph{Hyperproperties} are subsets of $2^\behaviorUnv$. A component $\cmp$
satisfies a hyperproperty $\cmpSet$ if $\cmp \in \cmpSet$. Since hyperproperties
allow us to define exactly what components satisfy them, we can define them
using any number of behaviors of a component (as opposed to trace properties
which can only predicate about single traces). 
One of our contributions is an assume-guarantee theory that supports
the expression of arbitrary hyperproperties. As we present our theory, we will
use the following running example.

\paragraph{Running example:}
Consider the
digital system shown in Figure \ref{fg:hbasak}; this system is similar to those
presented in \cite{Rabe_2016,10.1007/978-3-319-99725-4_17} to illustrate the
non-interference property in security. Here, we have an $s$-bit secret data input $S$ and an $n$-bit public
input $P$. The system has an output $O$. There is also an input $H$ that is
equal to zero when the system is being accessed by a user with low-privileges,
i.e., a user not allowed to use the secret data, and equal to one otherwise. We wish
the overall system to satisfy the property that for all environments with $H =
0$, the implementations can only make the output $O$ depend on $P$, the public
data, not on the secret input $S$. To see why a trace property cannot capture this requirement, suppose for simplicity that all variables are 1-bit-long. A trace property that refines the required non-interference property is
$P = \{(H= 0, P = 1, S = 1, O = 1), (H= 0, P = 0, S = 1, O = 0),
(H= 0, P = 1, S = 0, O = 1), (H= 0, P = 0, S = 0, O = 0)\}$. A valid implementation $M$ of $P$ is
$M = \{(H= 0, P = 1, S = 1, O = 1), (H= 0, P = 0, S = 0, O = 0)\}$, but the component $M$ leaks the value of $S$ in its output. We conclude that non-interference does not behave as a trace property.
In our development, we will use hypercontracts first to express
this top-level assume-guarantee requirement, and then to find a component that added to a
partial implementation of the system results in a design that meets the
top-level spec.

\emph{Non-interference}, introduced by Goguen and Meseguer \cite{6234468}, is a
common information-flow attribute, a prototypical
example of a design quality which trace properties are unable to
capture~\cite{clarkson2010hyperproperties}. It can be expressed with
hyperproperties, and is in fact one reason behind their introduction.

Suppose $\sigma$ is one of the behaviors that our system can display,
understood as the state of memory locations through time. Some of those memory locations we call \emph{privileged}, some \emph{unprivileged}. 
Let $L_0(\sigma)$ and $L_f(\sigma)$ be the projections of the behavior $\sigma$
to the unprivileged memory locations of the system, at time zero, and at the final time (when execution is done). We say that a component $M$ meets the non-interference hyperproperty when
\[
    \forall \sigma,
    \sigma' \in M.\,\, L_0(\sigma) = L_0(\sigma') \Rightarrow
    L_f(\sigma) = L_f(\sigma'),
\]
i.e., if two traces begin with the unprivileged locations in the same state, the final state of the unprivileged locations matches.

Non-interference is a downward-closed hyperproperty \cite{Rabe_2016,10.1007/978-3-319-99725-4_17}, and a $2$-safety hyperproperty---hyperproperties
called \emph{$k$-safety} are those for the refutation of which one must provide
at least $k$ traces. In our example, to refute the hyperproperty, it suffices to
show two traces that share the same unprivileged initial state, but which differ
in the unprivileged final state.

\begin{figure}[t]
    \begin{subfigure}{.31\columnwidth}
        \centering
        \includegraphics[width=0.9\columnwidth]{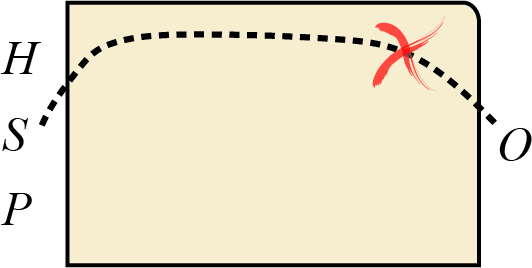}
        \caption{}
        \label{fg:hbasak}
    \end{subfigure}
    \begin{subfigure}{.31\columnwidth}
        \centering
        \includegraphics[width=0.9\columnwidth]{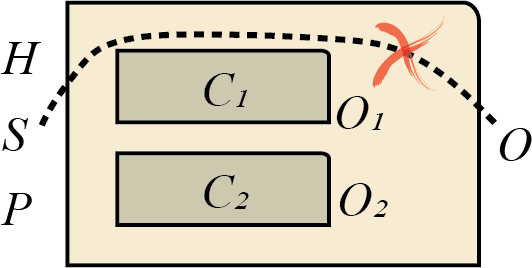}
        \caption{}
        \label{fg:alPM}
    \end{subfigure}
    \begin{subfigure}{.31\columnwidth}
        \centering
        \includegraphics[width=0.9\columnwidth]{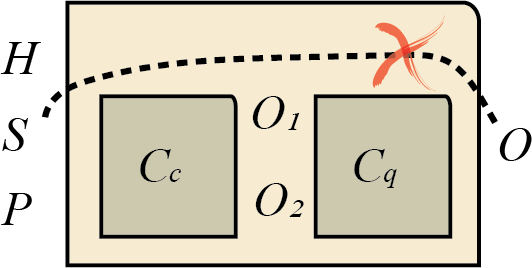}
        \caption{}
        \label{fg:kabhgoi}
    \end{subfigure}
    \caption{(a) A digital system with a secret input $S$ and a public input $P$. The overall
    system must meet the requirement that the secret input does not affect the value of the
    output $O$ when the signal $H$ is deasserted (this signal is asserted when a privileged
    user uses the system). Our agenda for this running example is the following: (b) we will start with two components satisfying hypercontracts $\cont_1$ and $\cont_2$ characterizing information-flow properties of their own;
    (c) their composition $\cont_c$ will be derived.
    Through the quotient $\cont_q$, we will discover the functionality that needs to be added in order
    for the design to meet the top-level information-flow spec $\cont$.}
    \label{nlakudbai}
\end{figure}

\paragraph{Abstract contract theories.} 
Since many contract frameworks have been proposed, there have also been efforts to
systematize this knowledge by building high-level theories of which existing
contract theories are instantiations. Thus, Bauer et al. \cite{specToContract}
describe how to build a contract theory if one has a specification theory
available. Benveniste et al. \cite{BenvenisteContractBook} provide a meta-theory
that builds contracts starting from an algebra of components. They provide
several operations on contracts and show how this meta-theory can describe,
among others, interface automata, assume-guarantee contracts, modal interfaces, and
rely-guarantee reasoning. This meta-theory is, however, low-level, specifying
contracts as unstructured sets of environments and implementations. As a
consequence, important concepts such as parallel composition and quotient of
contracts are expressed in 
terms that are considered too abstract---see~\cite{BenvenisteContractBook}, chapter 4. For example, no closed
form formula is given for the quotient besides its abstract definition as
adjoint of parallel composition.

\paragraph{Contributions.}
In this paper, we provide a theory of contracts, called \emph{hypercontracts},
that addresses the above deficiencies by requiring more structure in definition
of components, environments, and implementations. This additional structure does
not restrict the applicability of our theory, however. This theory is built in
three stages. We begin with a theory of components. Then we state what are the
sets of components that our theory can express; we call such objects \compsets,
which are equivalent to hyperproperties in behavioral
formalisms~\cite{10.1007/978-3-319-99725-4_17}. From these \compsets, we build
hypercontracts. We provide closed-form expressions for hypercontract
manipulations. Then we show how the hypercontract theory applies to two specific
cases: downward-closed hypercontracts and interface hypercontracts (equivalent
to interface automata). The main difference between hypercontracts and the
meta-theory of contracts is that hypercontracts are more structured: the
meta-theory of contracts defines a theory of components, and uses these
components in order to define contracts. Hypercontracts use the theory of
components to define compsets, which are the types of properties that we are
interested in representing in a specific theory. Hypercontracts are built out of
compsets, not out of components.

To summarize, our contributions are the following: $(i)$ a new model of \emph{hypercontracts}
possessing a richer algebra than the metatheory
of~\cite{BenvenisteContractBook} and capable of expressing any lattice of hyperproperties and $(ii)$ a calculus of \emph{\conic
hypercontracts} offering finite representations of downward-closed
hypercontracts.


\section{Preliminaries}

\paragraph{Preorders.} Many concepts in this paper will be inherited from preorders. We recall that a
preorder $(P, \le)$ consists of a set $P$ and a relation $\le$ which is
transitive (i.e., $a \le b$ and $b \le c$ implies that $a \le c$ for all $a, b, c
\in P$) and reflexive ($a \le a$ for all $a \in P$). A partial order is a
preorder whose relation is also antisymmetric (i.e., from $a \le b$ and $b \le
a$ we conclude that $a = b$). 

Our preorders will come equipped with a partial binary operation called
composition, usually denoted $\times$. Composition is often understood as a
means of connecting elements together and is assumed to be monotonic in the
preorder, i.e., we assume composing with bigger elements yields bigger results:
$
    \forall a, b , c \in P.\,\, a \le b \Rightarrow a\times c \le b \times c.
$
We will also be interested in taking elements apart. For a notion of
composition, we can always ask the question, for $a, b \in P$, what is
the largest element $b \in P$ such that $a \times b \le c$? Such an element is
called \emph{quotient} or \emph{residual}, usually denoted $c / a$. Formally,
the definition of the quotient $c / a$ is
\begin{align}\label{eq:quotientDef}
    \forall b \in P.\,\, a \times b \le c \text{ if and only if } b \le c/a,
\end{align}
which means that the quotient is the right adjoint of composition (in the sense
of category theory). A synonym of this notion is to say that composing by a
fixed element $a$ (i.e., $b \mapsto a \times b$) and taking quotient by the same
element (i.e., $c \mapsto c / a$) form a Galois connection. A description of the
use of the quotient in many fields of engineering and computer science is given
in \cite{EPTCS326.14}.


A partial order for which every two elements have a well-defined LUB (aka join),
denoted $\join$, and GLB (aka meet), denoted $\meet$, is a lattice. A lattice in
which the meet has a right adjoint is called Heyting algebra. This right adjoint
usually goes by the name exponential, denoted $\to$. In other words, the
exponential is the notion of quotient if we take composition to be given by the
meet, that is, for a Heyting algebra $H$ with elements $a, c$, the exponential
is defined as
\begin{align}\label{o8tibobno}
    \forall b \in H.\,\, a \meet b \le c \text{ if and only if } b \le a \to c,
\end{align}
which is the familiar notion of implication in Boolean algebras.

\section{The theory of hypercontracts}

Our objective is to develop a theory of assume-guarantee reasoning for any kind
of attribute of Cyber-Physical Systems. We do this in three steps: 
\begin{enumerate}
    \item  
we consider components coming with notions of preorder (e.g., simulation) and parallel composition; 
\item 
we discuss the notion of a \compset\ and give it substantial algebraic structure---unlike the unstructured sets of components considered in the metatheory of~\cite{BenvenisteContractBook};
\item
we build hypercontracts as pairs of \compsets with additional structure---capturing environments and implementations. 
\end{enumerate}
In this section we describe how this construction is performed, and in the next
we show specialized hypercontract theories.

\subsection{Components}

In the theory of hypercontracts, the most primitive concept is the component.
Let $(\cmpUniverse, \le)$ be a preorder. The elements $\cmp \in \cmpUniverse$
are called \emph{components}. We say that $\cmp$ is a subcomponent of $\cmp'$
when $\cmp \le \cmp'$. If we represented components as automata, the statement
``is a subcomponent of'' is equivalent to ``is simulated by.''

There exists a partial binary operation, $\cmpComposition: \cmpUniverse,
\cmpUniverse \to \cmpUniverse$, monotonic in both arguments, called
\emph{composition}. If $\cmp\cmpComposition\cmp'$ is not defined, we say that
$\cmp$ and $\cmp'$ are \textit{non-composable} (and \textit{composable} otherwise). A component
$\cmpEnv$ is an environment for component $\cmp$ if $\cmpEnv$ and $\cmp$ are
composable. We assume that composition is associative and commutative.

Similarly, we assume the existence of a second, partial binary operation that is
the right adjoint of composition: the \textit{quotient} \eqref{eq:quotientDef}
for the component theory. Given two components $\cmp$ and $\cmp'$, the quotient,
denoted $\cmp \cmpQuotient \cmp'$, is the largest component $\cmp''$ satisfying
$\cmp' \cmpComposition \cmp'' \le \cmp$. In other words, it gives us the largest
component whose composition with $\cmp'$ is a subcomponent of $\cmp$.

\smallskip
\paragraph{Running example.} In order to reason about possible decompositions of the system shown in Figure \ref{fg:hbasak}, we introduce the internal variables $O_1$ and $O_2$, as shown in Figure \ref{fg:alPM}. They have lengths $o_1$ and $o_2$, respectively. The output $O$ has length $o$.
For simplicity, we will assume that the behaviors of the entire system are stateless. In that case, the set of components $\cmpUniverse$ is the union of the following sets:
\begin{itemize}[wide=0pt]
    \item For $i \in \{1, 2\}$, components with inputs $H$, $S$, $P$, and output $O_i$, i.e., the sets $\{(H, S, P, O_1, O_2, O) \,\,|\,\, \exists f \in (2^1\times 2^s \times 2^n \to 2^{o_i}).\, O_i = f(H, S, P)\}$.
    \item Components with inputs $H$, $S$, $P$, $O_1$, $O_2$, and output $O$, i.e., the set $\{(H,\linebreak[1] S,\linebreak[1] P,\linebreak[1] O_1,\linebreak[1] O_2,\linebreak[1] O) \,\,|\,\, \exists f \in (2^1\times 2^s \times 2^n \times 2^{o_1} \times 2^{o_2} \to 2^o).\, O = f(H, S, P, O_1, O_2)\}$. We also consider components any subset of these components, as these correspond to restricting inputs to subsets of their domains.
\end{itemize}
In this theory of components, composition is carried out via set intersection. So for example, if for $i \in \{1,2\}$ we have functions $f_i \in (2^1\times 2^s \times 2^n \to 2^{o_i})$ and components $M_i = \{(H, S, P, O_1, O_2, O) \,\,|\,\, O_i = f_i(H, S, P)\}$, the composition of these objects is $$M_1 \cmpComposition M_2 = \left\{(H, S, P, O_1, O_2, O) \,\left| \begin{array}{r} O_1 = f_1(H, S, P) \\ O_2 = f_2(H, S, P)\end{array}\!\!\!\right.\right\}$$ which is the set intersection of the components's behaviors.


\subsection{\Compsets}

$\catCmpSet$ is a lattice whose objects are sets of components, called \emph{\compsets}. Thus, {compsets} are equivalent to hyperproperties when the underlying component theory represents components as sets of behaviors. 
In general, not every set of components is necessarily an object of $\catCmpSet$.

$\catCmpSet$ comes with a notion of satisfaction. Suppose $\cmp \in
\cmpUniverse$ and $\cmpSet$ is a \compset. We say that $\cmp$ \emph{satisfies}
$\cmpSet$ or conforms to $\cmpSet$, written $\cmp \models \cmpSet$, when $\cmp
\in \cmpSet$. For \compsets $\cmpSet, \cmpSet'$, we say that $\cmpSet$ \emph{refines}
$\cmpSet'$, written $\cmpSet \le \cmpSet'$, when $\cmp \models \cmpSet
\Rightarrow \cmp \models \cmpSet',$ i.e., when $\cmpSet \subseteq \cmpSet'$.

Since we assume $\catCmpSet$ is a lattice, the greatest lower bounds and least
upper bounds of finite sets are defined. Observe, however, that although the
partial order of $\catCmpSet$ is given by subsetting, the meet and join of
$\catCmpSet$ are not necessarily intersection and union, respectively, as the
union or intersection of any two elements are not necessarily elements of
$\catCmpSet$. 

\paragraph{Running example.} We are particularly interested in writing non-interference
specifications. Regarding the system shown in Figure \ref{fg:hbasak}, we require the top level component to generate the output $O$ independently from the secret input $S$. We build our theory of compsets by letting the set $2^{\cmpUniverse}$
be the set of elements of $\catCmpSet$. This means that any set of components is a
valid \compset. 
The components meeting the top-level non-interference property are those belonging to the compset $\{(H, S, P, O_1, O_2, O)\,\,|\,\, \exists f \in (2^1 \times 2^{n} \to 2^o).\, O =
f(H, P)\}$, i.e., those components for which $H$ and $P$ are sufficient to evaluate $O$.
This corresponds exactly to those components that are insensitive to the secret input $S$.
The join and meet of these compsets is given by set union and intersection, respectively.

\subsubsection{Composition and quotient}

We extend the notion of composition to $\catCmpSet$:
{\small
\begin{align}\label{eq:hypComp}\cmpSet \cmpComposition \cmpSet' =
    \setArg{\cmp \cmpComposition \cmp'}{
        \begin{aligned}
            & \cmp
\models \cmpSet \text{, } \cmp' \models \cmpSet' \text{, and} \\ &
\text{$\cmp$ and
$\cmp'$ are composable}\end{aligned}
}.\end{align}}
Composition is total and monotonic, i.e., if $\cmpSet' \le \cmpSet''$, then
$\cmpSet \cmpComposition \cmpSet' \le \cmpSet\cmpComposition\cmpSet''$. It is
also commutative and associative, by the commutativity and associativity,
respectively, of component composition.

We assume the existence of a second (but partial) binary operation on the
objects of $\catCmpSet$. This operation is the right adjoint of composition: for
\compsets $\cmpSet$ and $\cmpSet'$, the residual $\cmpSet\contQuotient\cmpSet'$
(also called \emph{quotient}), is defined by the universal property
\eqref{eq:quotientDef}. From the definition of composition, we must have
\begin{align}\label{eq:hypQuot}
    \cmpSet\contQuotient\cmpSet' = \setArg{\cmp \in \cmpUniverse}{ \{\cmp\} \cmpComposition \cmpSet' \subseteq \cmpSet }.
\end{align}

\subsubsection{Downward-closed \compsets}

The set of components was introduced with a partial order. We say that a
\compset $\cmpSet$ is \emph{downward-closed} when $\cmp' \le \cmp$ and $\cmp
\models \cmpSet$ imply $\cmp' \models \cmpSet$, i.e., if a component
satisfies a downward-closed \compset, so does its subcomponent. Section
\ref{sc:ixoiagbg} treats downward-closed \compsets in detail.

\subsection{Hypercontracts}
\myparagraph{Hypercontracts as pairs (environments, closed-system specification)}
A hypercontract is a specification for a design element that tells what is
required from the design element when it operates in an environment that meets
the expectations of the hypercontract. A hypercontract is thus a pair of
\compsets: \[\cont = (\envSet, \sysSet) =(\mbox{environments, closed-system specification}).\]
$\envSet$ states the environments in which the object being specified must
adhere to the specification. $\sysSet$ states the requirements that the design
element must fulfill when operating in an environment which meets the
expectations of the hypercontract. We say that a component $\cmpEnv$ \emph{is an
environment of hypercontract} $\cont$, written $\cmpEnv\satEnv\cont$, if
$\cmpEnv\models\envSet$. We say that a component $\cmp$ \emph{is an
implementation of} $\cont$, written $\cmp\satImpl\cont$, when $\cmp
\cmpComposition \cmpEnv \models \sysSet \text{ for all }\cmpEnv \models
\envSet.$ We thus define the set of implementations $\impSet$ of $\cont$ as the
\compset containing all implementations, i.e., as the quotient:
\[
 \mbox{implementations} =  \impSet = \sysSet \contQuotient \envSet.
\]
A hypercontract with a nonempty set of environments is called \emph{compatible};
if it has a nonempty set of implementations, it is called \emph{consistent}. For
$\sysSet$ and $\impSet$ as above, the \compset $\envSet'$ defined as $\envSet'=
\sysSet \contQuotient \impSet$ contains all environments in which the
implementations of $\cont$ satisfy the specifications of the hypercontract.
Thus, we say that a hypercontract is saturated if its environments \compset is as
large as possible in the sense that adding more environments to the
hypercontract would reduce its implementations. This means that $\cont$
satisfies the following fixpoint equation:
\[
    \envSet = \sysSet / \impSet = \sysSet / (\sysSet \contQuotient \envSet).
\]
\myparagraph{Hypercontracts as pairs (environments, implementations)}
Another way to interpret a hypercontract is by telling explicitly which
environments and implementations it supports. Thus, we would write the
hypercontract as $\cont = (\envSet, \impSet)$. We will see that assume-guarantee
theories can differ as to what is the most convenient representation for their
hypercontracts. 

\smallskip

\myparagraph{The lattice $\catContract$ of hypercontracts}
Just as with $\catCmpSet$, we define $\catContract$ as a lattice
formed by putting together two \compsets in one of the above two ways. Not every pair of
\compsets is necessarily a valid hypercontract. We will define soon the
operations that give rise to this lattice.

\subsubsection{Preorder}

We define a preorder on hypercontracts as follows: we say that $\cont$
\emph{refines} $\cont'$, written $\cont \le \cont'$, when every environment of
$\cont'$ is an environment of $\cont$, and every implementation of $\cont$ is an
implementation of $\cont'$, i.e., $\cmpEnv \satEnv \cont' \Rightarrow \cmpEnv
\satEnv \cont \text{ and }  \cmp \satImpl \cont \Rightarrow \cmp \satImpl
\cont'$. We can express this as
\begin{align*}
&\envSet' \le \envSet \text{ and } 
\sysSet \cmpQuotient \envSet = \impSet \le \impSet' = \sysSet' \cmpQuotient \envSet'.
\end{align*}
Any two $\cont, \cont'$ with $\cont \le \cont'$ and $\cont' \le \cont$ are said
to be \emph{equivalent} since they have the same environments and the same
implementations. We now obtain some operations using preorders which are defined
as the LUB or GLB of $\catContract$. We point out that the expressions we obtain
are unique up to the preorder, i.e., up to hypercontract equivalence.

\smallskip

\subsubsection{GLB and LUB} From the preorder just defined, the GLB of $\cont$
and $\cont'$ satisfies:
\mbox{$
    \cmp \satImpl \cont \meet \cont'$} { if and only if }
   $ \cmp \satImpl \cont \text{ and } \cmp \satImpl \cont'
$;
and
$
    \cmpEnv \satEnv \cont \meet \cont' \text{ if and only if }
    \cmpEnv \satEnv \cont \text{ or } \cmpEnv \satEnv \cont'
$.

Conversely, the least upper bound satisfies
$
    \cmp \satImpl \cont \join \cont' \text{ if and only if }
    \cmp \satImpl \cont \text{ or } \cmp \satImpl \cont'
$,
and
\mbox{$
    \cmpEnv \satEnv \cont \join \cont'$} { if and only if }
    $\cmpEnv \satEnv \cont$ { and } $\cmpEnv \satEnv \cont'.
$

The lattice $\catContract$ has hypercontracts for objects (up to contract equivalence), and meet and join as
just described.

\subsubsection{Parallel composition}\label{khabq187qb} The composition of
hypercontracts $\cont_i = (\envSet_i, \impSet_i)$ for $1 \le i \le n$,
denoted $\parallel_i \cont_i$, is the smallest hypercontract $\cont'
= (\envSet', \impSet')$ (up to equivalence) meeting the following
requirements:
\begin{itemize}[wide=0pt]
    \item any composition of implementations of all $\cont_i$ is an
    implementation of $\cont'$; and
    \item for any $1 \le j \le n$, any composition of an environment of $\cont'$
    with implementations of all $\cont_i$ (for $i \ne j$) yields
    an environment for $\cont_j$.
\end{itemize}
These requirements were stated for the first time by Abadi and Lamport
\cite{AbadiLamportComposingSpecs}. Using our notation, this composition principle becomes
{\footnotesize
\begin{align}
    \nonumber
    \cont & \contComposition \cont' = \bigwedge
    \setArg{\begin{aligned} & (\envSet', \impSet') \\ & \in \catContract \end{aligned}}{
    \left[\begin{array}{ll}
        \impSet_1 \cmpComposition \ldots \cmpComposition \impSet_n \le \impSet', \text{ and}\\
        \envSet' \cmpComposition \impSet_1 \cmpComposition \ldots \cmpComposition \hat \impSet_j \cmpComposition \ldots \cmpComposition \impSet_n  \le \envSet_j \\ 
        \text{for all } 1 \le j \le n
    \end{array}\right]} \\
    \label{eq:cntComp} &= 
    \bigwedge
    \setArg{\begin{aligned} & (\envSet', \impSet') \\ & \in \catContract \end{aligned}}{
    \left[\begin{array}{ll}
        \impSet_1 \cmpComposition \ldots \cmpComposition \impSet_n \le \impSet', \text{ and}\\
        \envSet' \le \bigwedge_{1 \le j \le n} \frac{\envSet_j}{\impSet_1 \cmpComposition \ldots \cmpComposition \hat \impSet_j \cmpComposition \ldots \cmpComposition \impSet_n}
    \end{array}\right]
    },
\end{align}
}where the notation $\hat \impSet_j$ indicates that the composition $\impSet_1 \cmpComposition \ldots \cmpComposition \hat \impSet_j \cmpComposition \ldots \cmpComposition \impSet_n$ includes all terms $\impSet_i$, except for $\impSet_j$.

\paragraph{Running example.}
Coming back to the example shown in Figure \ref{nlakudbai}, we want to state a requirement for the top-level component that for
all environments with $H = 0$, the implementations can only make the output $O$
depend on $P$, the public data. We will write a hypercontract for the top-level.
We let  $\cont = (\envSet, \impSet)$, where
\begin{align*}
    \envSet =& \setArg{M \in \cmpUniverse}{\forall (H, S, P, O_1, O_2, O) \in M.\, H = 0}
    \\
    \impSet =&  \{M \in \cmpUniverse \,\,|\,\, \exists f {\in} (2^n \to 2^o).  \forall (H, S, P, O_1, O_2, O) \in M.\, H = 0 \rightarrow O {=} f(P) \}.
\end{align*}
The environments are all those components only defined for $H = 0$. The implementations are those such that the output is a function of $P$ when $H = 0$.

Let $f^*: 2^n \to 2^o$. Suppose we have two hypercontracts that require their implementations to satisfy the function $O_i = f^*(P)$, one implements it when $S = 0$, and the other when $S \ne 0$. For simplicity of syntax, let $s_1$ and $s_2$ be the propositions $S = 0$ and $S \ne 0$, respectively.
Let the two hypercontracts be $\cont_i = (\envSet_i, \impSet_i)$ for $i \in \{1,2\}$. We won't place restrictions on the environments for these hypercontracts, so we obtain $\envSet_i = \cmpUniverse$ and
\begin{align*}
    \impSet_i = &\{M \in \cmpUniverse \,\,|\,\, \forall (H, S, P, O_1, O_2, O) \in M. s_i \rightarrow O_i {=} f^*(P) \}.
\end{align*}

We now evaluate the composition of these two hypercontracts: $\cont_c=\cont_1\contComposition \cont_2 = (\envSet_c \,,\, \impSet_c)$, yielding $\envSet_c = \cmpUniverse$ and
\begin{align*}
    \impSet_c = &\{M \in \cmpUniverse \,\,|\,\, \forall (H, S, P, O_1, O_2, O) \in M. \\ & (s_1 \rightarrow O_1 {=} f^*(P)) \land (s_2 \rightarrow O_2 {=} f^*(P)) \}.
\end{align*}

\subsubsection{Mirror or reciprocal}

We assume we have an additional operation on hypercontracts, called both mirror
and reciprocal, which flips the environments and implementations of a
hypercontract: $\cont^{-1} = (\envSet, \impSet)^{-1} = (\impSet, \envSet)$ and
$\cont^{-1} = (\envSet, \sysSet)^{-1} = (\sysSet / \envSet, \sysSet)$. This
notion gives us, so to say, the hypercontract obeyed by the environment. The
introduction of this operation assumes that for every hypercontract $\cont$, its
reciprocal is also an element of $\catContract$. Moreover, we assume that, when
the infimum of a collection of hypercontracts exists, the following identity
holds:
\begin{align}\label{eq:recipOptIdentity}\begin{array}{c}
      \recip{\left(\bigmeet_i \cont_i\right)} = \bigjoin_i \recip{\cont_i}.
\end{array}
\end{align}


\subsubsection{Hypercontract quotient}

The \emph{quotient} or residual for hypercontracts $\cont = (\envSet, \impSet)$ and
$\cont'' = (\envSet'', \impSet'')$, written $\cont'' \contQuotient \cont$, has
the universal property \eqref{eq:quotientDef}, namely $\forall \cont'.\,\, \cont
\contComposition \cont' \le \cont'' $ if and only if
$\cont' \le \cont'' \cmpQuotient \cont$.
We can obtain a closed-form expression using the reciprocal:

\begin{proposition}\label{iaxbisoudy} The hypercontract quotient obeys $\cont'' / \cont = \left( (\cont'')^{-1} \contComposition \cont
    \right)^{-1}$.
\end{proposition}

\paragraph{Running example.}
We use the quotient to find the specification of the component that we need to add to the system shown in Figure \ref{fg:kabhgoi} in order to meet the top level contract $\cont$. 
To compute the quotient, we use \eqref{eq:dcContQuotient}. 
We let $\cont \contQuotient \cont_c = (\envSet_q,\,\, \impSet_q)$ and obtain $\envSet_q = \envSet  \land \impSet_c$ and
\begin{align*}
    \impSet_q =& \{M \in \cmpUniverse \,\,|\,\, \exists f \in (2^n \to 2^o) \forall (H, S, P, O_1, O_2, O) \\ & \in M. \left((s_1 \rightarrow O_1 {=} f^*(P)) \land (s_2 \rightarrow O_2 {=} f^*(P))\right) \rightarrow
    O {=} f(P) \}.
\end{align*}
We can refine the quotient by lifting any restrictions on the environments, and
picking from the implementations the term with $f = f^*$. Observe that $f^*$ is a valid choice for $f$. This yields the
hypercontract $\cont_3 = (\envSet_3, \impSet_3)$, defined as $\envSet_3 = \cmpUniverse$ and
\begin{align*}
    \impSet_3 = &\{M \in \cmpUniverse \,\,|\,\, \forall (H, S, P, O_1, O_2, O) \in M. \\ & \left((s_1 \rightarrow O_1 {=} f^*(P)) \land (s_2 \rightarrow O_2 {=} f^*(P))\right) \rightarrow
    O {=} f^*(P) \}.
\end{align*}
A further refinement of this hypercontract is $\cont_r = (\envSet_r, \impSet_r)$, where $\envSet_r = \cmpUniverse$ and
\begin{align*}
    \impSet_r = &\{M \in \cmpUniverse \,\,|\,\, \forall (H, S, P, O_1, O_2, O) \in M. \left((s_1 \rightarrow O {=} O_1) \land (s_2 \rightarrow O {=} O_2)\right) \}.
\end{align*}
By the properties of the quotient, composing
this hypercontract, which knows nothing about $f^*$, with $\cont_c$ will yield a
hypercontract which meets the non-interference hypercontract $\cont$. Note that
this hypercontract is consistent, i.e., it has implementations (in general,
refining may lead to inconsistency).

\subsubsection{Merging}

The composition of two hypercontracts yields the specification of a system
comprised of two design objects, each adhering to one of the hypercontracts
being composed. Another important operation on hypercontracts is viewpoint
merging, or \emph{merging} for short. It can be the case that the same design
element is assigned multiple specifications corresponding to multiple
viewpoints, or design concerns \cite{Benveniste2008,agContractMerging} (e.g.,
functionality and a performance criterion). Suppose $\cont_1 = (\envSet_1,
\sysSet_1)$ and $\cont_2 = (\envSet_2, \sysSet_2)$ are the hypercontracts we
wish to merge. Two slightly different operations can be considered as candidates for formalizing viewpoint
merging:
\begin{itemize}[wide=0pt]
    \item A \emph{weak merge} which is the GLB; and
     \item A \emph{strong merge} which states that environments of the merger
     should be environments of both $\cont_1$ and $\cont_2$ and that the closed
     systems of the merger are closed systems of both $\cont_1$ and $\cont_2$.
     If we let $\cont_1\bullet\cont_2 = (\envSet, \impSet)$, we have
\end{itemize}
\begin{align*}
    \envSet &= \begin{aligned} \join \{ \envSet' \in \catCmpSet \,\, | \,\, \envSet' \le \envSet_1
    \meet \envSet_2 \text{ and }  
    \exists\, \cont'' = (\envSet'', \impSet'') \in
    \catContract.\,\, \envSet' = \envSet'' \} \end{aligned} \\
    \impSet &= \join \left\{\impSet' \in
    \catCmpSet \,\,\middle |\,\, \begin{aligned} & \impSet' \le (\sysSet_1 \meet \sysSet_2)/\envSet
    \text{ and } \\ & (\envSet, \impSet) \in \catContract \end{aligned} \right\}.
\end{align*}
The difference is that, whereas the commitment to satisfy $\sysSet_2$ survives
when under the weak merge when the environment fails to satisfy $\envSet_1$, no
obligation survives under the strong merge. This distinction was proposed
in~\cite{DBLP:journals/ejcon/Sangiovanni-VincentelliDP12} under the name of
weak/strong assumptions.

\section{Representation of compsets and hypercontracts}
\newcommand{\den}[0]{\text{Den}}

We have laid out the theory of hypercontracts, built in three stages. We now
discuss the issue of syntactically representing these objects. Up to now, we
have written compsets explicitly as sets. Doing this, however, results in a
problem of \emph{portability}. Consider again the example shown in Figure
\ref{nlakudbai}. In our running example, we found out that we could express the
property of non-interference for the top-level component through the expression
$\{(H, S, P, O_1, O_2, O)\,\,|\,\, \exists f \in (2^1 \times 2^{n} \to 2^o).\, O
= f(H, P)\}$. What would happen if we added more internal variables to the
system? Suppose, for example, that we have an additional variable $O_3$. In that
case, the theory of components needs to define component behaviors also over the
variable $O_3$, and the compset in question becomes $\{(H, S, P, O_1, O_2, O_3,
O)\,\,|\,\, \exists f \in (2^1 \times 2^{n} \to 2^o).\, O = f(H, P)\}$. This
makes it clear that compsets change when the theory
of components modifies its variables. Yet, we would agree that the two compsets we wrote represent the
same components.

In order to have a representation of compsets which is invariant to adding new
variable names to the theory of components, we assume we have a logic $\Psi$
whose formulas are denoted by compsets. We require $\Psi$ to be a lattice and
the denotation map $$\den: \Psi \to \catCmpSet$$ to be a lattice map. This means
that $\den(\psi \land \psi') = \den(\psi) \land \den (\psi')$ and $\den(\psi
\lor \psi') = \den(\psi) \lor \den (\psi')$. The $\den$ map also provides us
with the means to represent hypercontracts, as these are given by a pair of
compsets.

\paragraph{Example}. As an example, suppose we have a theory with only one
component: a voltage amplifier with an output $O$ having the same real value as
its input $I$. The component is given by $M = \{(I, O) \in \mathbb{R}^2 \,\,|\,\,
O = I\}$. The theory of compsets has two elements: $\emptyset$ and $\{M\}$.
Suppose we have a logic $\Psi$ with symbols $i, o$ in which the formula $\psi
\coloneqq i = o$ is well defined and has a denotation $\den(\psi) = \{C \in
\cmpUniverse \,\,|\,\, \forall (I,O) \in C.\, I = O\} = \{M\}$.

Now suppose we alter the component theory so that it has an additional real
variable $T$. Now the component $M$ becomes $M' = \{(I, O, T) \in \mathbb{R}^3
\,\,|\,\, O = I\}$. Observe that the description of the component $M$ has
changed; yet, we could say that $M'$ is completely independent of $T$. Now
suppose we have a logic $\Psi'$ with symbols $i, o, t$ in which the formula
$\psi \coloneqq i = o$ is also well-defined. We can build a denotation map
$\den' : \Psi' \to \catCmpSet$ such that $\den'(\psi) = \{C \in \cmpUniverse
\,\,|\,\, \forall (I,O, T) \in C.\, I = O\} = \{M'\}$.

We observe in this example that we were able to use the same formula $\psi$ in
order to represent a compset, even when we modified the underlying symbols on
which objects were defined. In other words, representations allow us to
define compsets by only using ``local knowledge'' about the interfaces of the
components described by the compset, despite the fact that components are
denoted on the set of behaviors of the entire system.

\section{Behavioral modeling}

In the behavioral approach to system modeling, we
start with a set $\behaviorUnv$ whose elements we call behaviors. Components are
defined as subsets of $\behaviorUnv$. They contain the behaviors they can
display. A component $\cmp$ is a subcomponent of $\cmp'$ if $\cmp'$ contains all
the behaviors of $\cmp$, i.e., if $\cmp \subseteq \cmp'$. Component composition
is given by set intersection: $\cmp \times \cmp' \defined \cmp
\setint \cmp'$. If we represent the components as $\cmp = \setArg{b \in
\behaviorUnv}{\phi(b)}$ and $\cmp' = \setArg{b \in \behaviorUnv}{\phi'(b)}$ for
some constraints $\phi$ and $\phi'$, then composition is $\cmp \times \cmp' =
\setArg{b \in \behaviorUnv}{\phi(b) \land \phi'(b)}$, i.e., the behaviors that
simultaneously meet the constraints of $\cmp$ and $\cmp'$. This notion of
composition is independent of the connection topology: the topology is inferred
from the behaviors of the components. The quotient is given by
implication: $\cmp \cmpQuotient \cmp' = \cmp' \to \cmp$.

We will consider three contract theories we can build with these components. The
first is based on unconstrained hyperproperties; the second is based on
downward-closed hyperproperties; and the third corresponds to assume-guarantee
contracts.

\subsection{General hypercontracts}

The most expressive behavioral theory of hypercontracts is obtained when we
place no restrictions on the structure of \compsets and hypercontracts. In this
case, the elements of $\catCmpSet$ are all objects $H \in 2^{2^{\behaviorUnv}}$,
i.e., all hyperproperties. The meet and join of \compsets are set intersection
and union, respectively, and their composition and quotient are given by
\eqref{eq:hypComp} and \eqref{eq:hypQuot}, respectively.
Hypercontracts are of the form $\cont = (\envSet, \impSet)$ with all extrema
achieved in the binary operations, i.e., for a second hypercontract $\cont' =
(\envSet', \impSet')$, the meet, join, and
composition \eqref{eq:cntComp} are, respectively, 
$\cont \meet \cont' = (\envSet \setunion
\envSet', \impSet \setint \impSet')$, 
$\cont \join \cont' = (\envSet \setint
\envSet', \impSet \setunion \impSet')$, and $\cont \contComposition \cont' =
\left(\frac{\envSet'}{\impSet} \setint \frac{\envSet}{\impSet'}, \impSet
\cmpComposition \impSet'\right)$.
From these follow the operations of quotient,
and merging.

\subsection{\Conic (or downward-closed) hypercontracts}
\label{sc:ixoiagbg}

We assume that $\catCmpSet$ contains exclusively downward-closed hyperproperties.
Let $\cmpSet \in \catCmpSet$. We say that $\cmp \models \cmpSet$ is a maximal
component of $\cmpSet$ when $\cmpSet$ contains no set bigger than $\cmp$, i.e.,
if
\[\forall \cmp' \models \cmpSet.\,\, \cmp \le \cmp' \Rightarrow \cmp' = \cmp.\]
We let $\maxcomp \cmpSet$ be the set of maximal components of $\cmpSet$:
\[
    \maxcomp \cmpSet = \setArg{\cmp \models \cmpSet}{ \forall \cmp' \models \cmpSet.\,\, \cmp \le \cmp' \Rightarrow \cmp' = \cmp}.
\]
Due to the fact $\cmpSet$ is downward-closed, the set of maximal components is a
unique representation of $\cmpSet$. We can express $\cmpSet$ as
\[\begin{array}{c}
    \cmpSet = \bigcup_{\cmp \in \maxcomp \cmpSet} \; 2^{\cmp}.
\end{array}
\]
We say that $\cmpSet$ is
\emph{$k$-\maximal}
if the cardinality of $\maxcomp \cmpSet$
is finite and equal to $k$, and we write this 
\[\mbox{$H=\langle \cmp_1, \ldots, \cmp_k\rangle$,~ where $\maxcomp \cmpSet=\{\cmp_1, \ldots, \cmp_k\}$.}\]
\subsubsection{Order}
The notion of order on $\catCmpSet$ can be expressed using
this notation as follows: suppose $\cmpSet' = \langle \cmp' \rangle_{\cmp' \in
\maxcomp \cmpSet '}$. Then 
\[
    \cmpSet' \le \cmpSet \text{ if and only if } \forall \cmp' \in \maxcomp \cmpSet ' \,\, \exists \cmp \in \maxcomp \cmpSet.\,\, \cmp' \le \cmp.
\]
\subsubsection{Composition}
Composition in $\catCmpSet$ becomes
\begin{align}\label{eq:dcCmpSetProduct}\begin{array}{c}
    \cmpSet \times \cmpSet' = \bigcup_{\substack{\cmp \in \maxcomp \cmpSet \\ \cmp' \in \maxcomp \cmpSet'}} \; 2^{\cmp \setint \cmp'} = \langle \cmp \setint \cmp' \rangle_{\substack{\cmp \in \maxcomp \cmpSet \\ \cmp' \in \maxcomp \cmpSet'}}.
    \end{array}
\end{align}
Therefore, if $\cmpSet$ and $\cmpSet'$ are, respectively, $k$- and $k'$-\maximal,
$\cmpSet \times \cmpSet'$ is at most $k k'$-\maximal.

\subsubsection{Quotient}
Suppose $\cmpSet_q$ satisfies
\[
    \cmpSet' \times \cmpSet_q \le \cmpSet.
\]
Let $\cmp_q \in \maxcomp \cmpSet_q$. We must have
\[
    \cmp_q \times \cmp' \models \cmpSet \text{ for every } \cmp' \in \maxcomp \cmpSet',
\]
which means that for each $\cmp' \in \maxcomp \cmpSet'$ there must exist an $\cmp
\in \maxcomp \cmpSet$ such that $\cmp_q \times \cmp' \le \cmp$; let us denote by $M(M')$ a choice $M'\mapsto{M}$ satisfying this condition. Therefore, we have
\begin{align}\label{ghqoudya}\begin{array}{c}
    \cmp_q \le \bigwedge_{\cmp' \in \maxcomp \cmpSet'} \;\frac{\cmp({\cmp'})}{\cmp'},
    \end{array}
\end{align}
Clearly, the largest such $\cmp_q$ is
obtained by making \eqref{ghqoudya} an equality. Thus, the cardinality of the
quotient is bounded from above by $k^{k'}$ since we have
\begin{align}\label{eq:dcCmpSetQuotient}\begin{array}{c}
    \cmpSet_q = \left\langle \bigwedge_{\cmp' \in \maxcomp \cmpSet'} \; \frac{\cmp({\cmp'})}{\cmp'} \right\rangle_{\substack{\cmp({\cmp')} \in \maxcomp \cmpSet \\ \forall \cmp' \in \maxcomp \cmpSet'}}.
    \end{array}
\end{align}

\subsubsection{Contracts}

Now we assume that the objects of $\catCmpSet$ are pairs of \emph{downward-closed \compsets.}
If we have two hypercontracts $\cont = (\envSet, \impSet)$ and $\cont' =
(\envSet', \impSet')$, their composition is
\begin{align}\label{eq:dcContProduct}
   \cont \contComposition \cont' = \left( \frac{\envSet}{\impSet'} \meet \frac{\envSet'}{\impSet}, \impSet \times \impSet' \right).
\end{align}
%
We can also write an expression for the quotient of two hypercontracts:
\begin{align}\label{eq:dcContQuotient}
   \cont / \cont' = \left( \envSet \times \impSet', \frac{\impSet}{\impSet'} \meet \frac{\envSet'}{\envSet}\right).
\end{align}



\subsection{Interval AG contracts}
\label{xknskaJK}

Now we explore AG with a must modality.
We will assume that elements of $\catCmpSet$ are property intervals. In other
words, if $\cmpSet$ is a \compset, we can find components $L, R \in
\cmpUniverse$ such that $\cmpSet = \setArg{\cmp \in \cmpUniverse}{L \le \cmp \le
R}$. We will refer to such \compsets as \emph{modal} or \emph{interval}
\compsets, and we write them as $\cmpSet = [L, R]$. The name modal is used to
indicate that a component satisfying a modal \compset must implement some
behaviors (those contained in $L$) and is only allowed to implement certain
behaviors (those contained in $R$).

Let $\cmpSet = [L ,R]$ and $\cmpSet'= [L',R']$. The operations on \compsets are given by
\begin{align*}
    \cmpSet \cmpComposition \cmpSet' =& \setArg{M \cmpComposition M'}{L \le M \le R \text{ and } L' \le M' \le R'}  =
    [L \setint L', R\setint R'], \\
    \cmpSet \meet \cmpSet' =& \setArg{M }{L \le M \le R \text{ and } L' \le M \le R'}  =
    [L \setunion L', R\setint R'],\\
    \cmpSet \join \cmpSet' =& \setArg{M }{L \le M \le R \text{ or } L' \le M \le R'}  =
    [L \setint L', R\setunion R'], \text{ and}\\
    \cmpSet \cmpQuotient \cmpSet' =& \join \setArg{ [L'', R''] }{ \cmpSet' \cmpComposition [L'', R''] \le \cmpSet } \\ =&
    \join \setArg{ [L'', R''] }{ [L' \setint L'', R' \setint R''] \le \cmpSet } \\ 
    =& [L, R \setunion \neg R'] \quad \text{(only defined when $L \le L'$)}.
\end{align*}

We now state the expressions for composition and quotient.

\begin{proposition}\label{uiqtdvo}
    Suppose $\cont = (\envSet, \impSet)$ and $\cont' = (\envSet', \impSet')$ with $\envSet = [L_e, R_e]$, $\impSet = [L_i, R_i]$, $\envSet' = [L_e', R_e']$, and $\impSet' = [L_i', R_i']$. The composition of these hypercontracts is only defined when $L_e = L_i = L_e' = L_i'$. Set $L = L_e$. Then the composition $\cont \contComposition \cont' = (\envSet_c, \impSet_c)$ is of the form
    \begin{align*}
        \envSet_c &=
    [L, (R_e \setint R_e') \setunion (R_e' \setint \neg R_i') \setunion (R_e \setint \neg R_i) ]\text{ and}\\
        \impSet_c &=
    [L, (R_i \setunion \neg R_e) \setint (R_i' \setunion \neg R_e') ].
    \end{align*}
    Now suppose $\cont = (\envSet, \impSet)$ and $\cont'' = (\envSet'', \impSet'')$ with $\envSet'' = [L_e'', R_e'']$, $\impSet'' = [L_i'', R_i'']$, $\envSet = [L_e, R_e]$, and $\impSet = [L_i, R_i]$. The residual $\cont'' \contQuotient \cont = (\envSet_r, \impSet_r)$ is only defined when $L_i'' \le L_i = L_e \le L_e''$. Call $L = L_i$. The components of the quotient have the form
    \begin{align*}
        \envSet_r =& [L, R_e'' \setint \left( R_i \setunion \neg R_e \right)] \text{ and} \\
        \impSet_r=&
    [L , \left(R_e \setint R_i''\right) \setunion \neg R_e'' \setunion \left( R_e \setint \neg R_i \right)].
    \end{align*}
\end{proposition}


\section{Receptive languages and interface hypercontracts}

In this section we connect the notion of a hypercontract with specifications
expressed as interface automata \cite{alfaroHenzingerIntAutomata}. 
With interface theories, we bring in the notion of
input-output profiles as an extra typing for components---so far, this was not considered in our development. This effectively partitions $\cmpUniverse$ into sets containing components sharing the same profile.

Our theory of components is constructed from a new notion called \emph{receptive
languages}. These objects can be understood as the trace denotations of
receptive I/O automata \cite{Lynch89anintroduction}. We will consider
downward-closed, 1-\maximal \compsets, see Section~\ref{sc:ixoiagbg}. And interface hypercontracts will be pairs of these
with a very specific structure. At the end of the section we show how the
denotation of interface automata is captured by interface hypercontracts. One
novelty of our approach is that the computation of the composition of
hypercontracts, which matches that of interface automata (as we will see), is
inherited from our general theory by specializing the component and \compset operations.

\subsection{The components are receptive languages}

Fix once and for all an alphabet $\alphabet$. When we operate on words of
$\alphabet^*$, we will use $\concat$ for word concatenation, and we'll let
$\pre(w)$ be the set of prefixes of a word $w$. These operations are extended to
languages: $$L \concat L' = \setArg{w \concat w'}{w \in L \text{ and } w' \in
L'},$$ and $\pre(L) = \bigcup_{w \in L} \pre(w)$. 
An \emph{input-output} signature of $\alphabet$ (or simply an IO signature when
the alphabet is understood), denoted $(I,O)$, is a partition of $\alphabet$ in
sets $I$ and $O$, i.e., $I$ and $O$ are disjoint sets whose union is
$\alphabet$.

\begin{definition}
Let $(I,O)$ be an IO signature. A language $L$ of $\alphabet$ is an
$I$-receptive language if 
\begin{itemize}[wide=0pt]
    \item $L$ is prefix-closed; and
    \item if $w \in L$ and $w' \in I^*$ then $w \concat w' \in L$.
\end{itemize}
The set of all $I$-receptive languages is denoted $\langSet_{I}$.
\end{definition}

\begin{proposition}\label{kbx8we7d} Let $(I,O)$ be an IO signature. Then
    $\langSet_{I}$ is closed under intersection and union.
\end{proposition}
Under the subset order, $\langSet_I$ is a lattice with intersection as the meet
and union as the join. Further, the smallest and largest elements of
$\langSet_I$ are, respectively, $0 = I^*$ and $1 = \alphabet^*$. It so happens
that $\langSet_I$ is a Heyting algebra. To prove this, it remains to be shown
that it has exponentiation (i.e., that the meet has a right adjoint).

\begin{proposition}\label{kjgh9yrbp} Let $L, L' \in \langSet_I$. The object
    $$
    L' \rightarrow L = \setArg{w \in \alphabet^*}{\pre(w) \setint L' \subseteq L}
    $$
    is an element of $\langSet_I$ and satisfies the property
    \eqref{o8tibobno} of the exponential.
\end{proposition}

We further explore the structure of the exponential. To do this, it will be
useful to define the following set: for languages $L, L'$ and a set $\Gamma
\subseteq \alphabet$, we define the set of \emph{missing $\Gamma$-extensions of
$L'$ with respect to $L$} as
$$
    \missingExt(L, L', \Gamma) = \left( \left((L \setint L') \concat \Gamma\right) \setminus L'\right) \concat \alphabet^*\,.
$$
The elements of this set are all words of the form $w \concat \sigma \concat
w'$, where $w \in L \setint L'$, $\sigma \in \Gamma$, and $w' \in \alphabet^*$.
These words satisfy the condition $w \concat \sigma \not \in L'$. In other
words, we find the words of $L \setint L'$ which, when extended by a symbol of
$\Gamma$, leave the language $L'$, and extend these words by the symbols that
make them leave $L'$ and then by every possible word of $\alphabet^*$.

\begin{proposition}\label{98wninhiu} Let $L, L' \in \langSet_I$. The exponential
    is given by $L' \rightarrow L = L \setunion \missingExt(L, L', O)$.
\end{proposition}

At this point, it has been established that each $\langSet_I$ is a Heyting
algebra. Now we move to composition and quotient, which involve languages of different IO signatures.

\subsection{Composition and quotient of receptive languages}

To every $I \subseteq \alphabet$, we have associated the set of languages
$\langSet_I$. Suppose $I' \subseteq I$. Then $L \in
\langSet_I$ if it is prefix-closed, and the extension of any word of $L$ by any
word of $I^*$ remains in $L$. But since $I' \subseteq I$, this means that the
extension of any word of $L$ by any word of $(I')^*$ remains in $L$, so $L \in
\langSet_{I'}$. We have shown that $I \subseteq I' \Rightarrow \langSet_{I'} \le
\langSet_{I}$. Thus, the map $I \mapsto \langSet_I$ is a contravariant functor
$2^\alphabet \to 2^{2^{\alphabet^*}}$.

Since $I' \subseteq I$ implies that $\langSet_{I} \le \langSet_{I'}$, we define
the embedding $\iota: \langSet_{I} \to \langSet_{I'}$ which maps a language of
$\langSet_{I}$ to the same language, but interpreted as an element of
$\langSet_{I'}$, 

Let $(I,O)$ and $(I', O')$ be IO signatures of $\alphabet$, $L \in \langSet_I$,
and $L' \in \langSet_{I'}$. The composition of structures with labeled inputs
and outputs traditionally requires that objects to be composed can't share
outputs. We say that IO signatures $(I,O)$ and $(I', O')$ are \emph{compatible}
when $O \setint O' = \emptyset$. This is equivalent to requiring that $I
\setunion I' = \alphabet$. Moreover, the object generated by the composition
should have as outputs the union of the outputs of the objects being composed.
This reasoning leads us to the definition of composition:

\begin{definition}[composition]
    Let $(I,O)$ and $(I', O')$ be compatible IO signatures of $\alphabet$. Let
    $L \in \langSet_I$ and $L' \in \langSet_{I'}$. The operation of language
    \emph{composition}, $\times: \langSet_{I}, \langSet_{I'} \to \langSet_{I \setint
    I'}$, is given by
    \[
        L \times L' = \iota L \land \iota' L',
    \]
    for the embeddings $\iota: \langSet_{I} \to \langSet_{I \setint I'}$ and
    $\iota': \langSet_{I'} \to \langSet_{I \setint I'}$.
\end{definition}
The adjoint of this operation is the quotient. We will investigate when the
quotient is defined. Let $I,I' \subseteq \alphabet$ with $I \subseteq I'$, $L
\in \langSet_I$, and $L' \in \langSet_{I'}$. Suppose there is $I_r \subseteq
\alphabet$ such that the composition rule $\times: \langSet_{I'}, \langSet_{I_r}
\to \langSet_{I}$ is defined. This means that $I' \setunion I_r = \alphabet$ and
$I' \setint I_r = I$. Solving yields $I_r = I \setunion \neg I' = I \setunion
O'$.

Observe that the smallest element of $\langSet_{I_r}$ is $I_r^*$. Thus, the
existence of a language $L'' \in \langSet_{I_r}$ such that $L'' \times L' \le L$
requires that $L' \setint I_r^* \subseteq L$. Clearly, not every pair $L, L'$
satisfies this property since we can take, for example, $L = I^*$ and $L' =
\alphabet^*$ to obtain $L' \setint I_r^* = (I \setunion O')^* \not \subseteq
I^*$, provided $I' \ne \alphabet$.

We proceed to obtain a closed-form expression for the quotient, but first we
define a new operator. For languages $L, L'$ and sets $\Gamma, \Delta \subseteq
\alphabet$, the following set of \emph{$(L', \Gamma, \Delta)$-uncontrollable extensions of $L\cap{L'}$}
\begin{align}
    \nonumber
    &\uncont(L, L', \Gamma, \Delta) = \\ &\setArg{\begin{aligned}w \in \\ L {\setint} L' \end{aligned}}{
        \begin{aligned}
        &\exists w' \in (\Gamma {\setunion} \Delta)^* \;\land\; \sigma {\in} \Gamma. \\
        &w {\concat} w' \in L {\setint} L' \;\land\; \\ &w {\concat} w' {\concat} \sigma \in L' {\setminus} L
        \end{aligned}}
        \concat \alphabet^*.
        \label{Uncont}
\end{align}
contains: $(i)$ all words of $L
\setint L'$ which can be uncontrollably extended to a word of $L' \setminus L$
by appending a word of $(\Gamma \setunion \Delta)^*$ and a symbol of $\Gamma$, and $(ii)$ all suffixes of such words. Equivalently, $\uncont(L, L', \Gamma, \Delta)$ contains all extensions of
the words $w \in L \setint L'$ such that there are extensions of $w$ by words
$w' \in (\Gamma \setunion \Delta)^*$ that land in $L'$ but not in $L$ after
appending to the extensions $w \circ w'$ a symbol of $\Gamma$. 

\begin{proposition}\label{nha9dbq9o} Let $(I,O)$ and $(I', O')$ be IO signatures
    of $\alphabet$ such that $I \subseteq I'$. Let $L \in \langSet_I$ and $L'
    \in \langSet_{I'}$. Let $I_r = I \setunion O'$, and assume that $L' \setint
    I_r^* \subseteq L$. Then the largest $L'' \in \langSet_{I_r}$ such that $L''
    \cmpComposition L' \le L$ is denoted $L / L'$ and is given by
    \begin{align*}
        L / L' =& \left(L \setint L' \setunion \missingExt\left(L, L', O' \right) \right) \setminus \uncont(L, L', O', I)
        .
    \end{align*}
\end{proposition}

We have defined receptive languages together with a preorder and a composition
operation with its adjoint. These objects will constitute our theory of
components, i.e., $\cmpUniverse = \oplus_{I \in 2^{\alphabet}} \langSet_I$.

\subsection{\Compsets and interface hypercontracts}
\label{bstxiauybo}

Using the set of components just defined, we proceed to build \compsets and
hypercontracts. The \compsets  contain components adhering to the
same IO signature. Thus, again the notion of an IO signature will partition the
set of \compsets (and the same will happen with hypercontracts). This means that for every
\compset $\cmpSet$, there will always be an $I \subseteq \alphabet$ such that
$\cmpSet \subseteq \langSet_I$.

For $I \subseteq \alphabet$, and $L \in \langSet_I$, we will consider \compsets
of the form 
$$\setArg{M \in 2^L}{I^* \subseteq M}, \mbox{ denoted by } [0 , L]\,,$$
where $0$
is $I^*$, the smallest element of $\langSet_I$,
i.e., the \compsets are all $I$-receptive
languages smaller than $L$.
We will focus on hypercontracts
whose implementations have signature $(I,O)$ and whose environments have
$(O,I)$. Thus, hypercontracts will consist of pairs $\cont = (\envSet, \sysSet)$
of $O$- and $\emptyset$-receptive \compsets, respectively. We will let 
$$\sysSet
= [0, S] = \setArg{M \in \langSet_{\emptyset}}{M \subseteq S}$$
for some $S \in
\langSet_{\emptyset}$. We will restrict the environments $E \in \envSet$ to
those that never extend a word of $S$ by an input symbol that $S$ does not
accept. The largest such environment is given by \begin{align}\label{8b7tqotbso}E_S = S \setunion
\missingExt(S, S, O).\end{align} Since $S$ is prefix-closed, so is $E_S$. Moreover,
observe that $E_S$ adds to $S$ all those strings that are obtained by
continuations of words of $S$ by an output symbol that $S$ does not produce.
This makes $E_S$ $O$-receptive. The set of environments is thus $\envSet = [O^*,
E_S]$.

Having obtained the largest environment, we can find the implementations. These
are given by $\impSet = [I^*, M_S]$ for $M_S = S / E_S$. Plugging the
definition, we have
\begin{align*}
    M_S &= \left(S \setint E_S \setunion \missingExt(S, E_S, I) \right) \setminus \uncont(S, E_S, I, \emptyset).
\intertext{There is no word of $I^*$ which can extend a word of $S$ into $E_S \setminus S$. Thus,}
    M_S &=
    S \setunion \missingExt(S, S, I)
    .
\end{align*}
Observe that $S$ and $\missingExt(S, S, I)$ are disjoint (same for
$\missingExt(S,S,O)$). Thus, $E_S \times M_S = S$. In summary, we observe that
our hypercontracts are highly structured. They are in 1-1 correspondence with a
language $S \in \langSet_\emptyset$ and an input alphabet $I \subseteq
\alphabet$, i.e., there is a set isomorphism 
\begin{equation}
\langSet_\emptyset, 2^\alphabet \xrightarrow{\sim} \catContract\,.
\label{usydtfsuy}
\end{equation}
Indeed, given $S$ and $I$, we build $E_S$ by
extending $S$ by $\alphabet \setminus I = O$, and $M_S$ by extending $S$ by $I$.
After this, the hypercontract has environments, closed systems, and
implementations $[O^*, E_S]$, $[\emptyset, S]$, and $[I^*, M_S]$, respectively.

\subsection{Hypercontract composition}
Let $S, S' \in \langSet_{\emptyset}$. We consider the composition of the
interface hypercontracts 
$
\cont_R = \cont_S \contComposition \cont_{S'}$, where $\cont_S = ([0, E_S], [0, S])$,
$\cont_{S'} = ([0, E_{S'}], [0, S'])$
and $E_S$ and $E_{S'}$ have signatures $(O,I)$ and $(O', I')$,
respectively. From the structure of interface hypercontracts, we have the
relations
\begin{align*}
    E_S &= S \setunion \missingExt(S, S, O)    \text{ and } \\
    E_{S'} &= S' \setunion \missingExt(S', S', O').
\end{align*}
Moreover, the implementations of $C, C'$ are, respectively, $\impSet = [I^*,
M_S]$ and $\impSet' = [I'^*, M_{S'}]$, where
\begin{align*}
    M_S &= S \setunion \missingExt(S, S, I)    \text{ and } \\
    M_{S'} &= S' \setunion \missingExt(S', S', I').
\end{align*}
The composition of these hypercontracts is defined if $(I, O)$ and $(I', O')$
have compatible signatures. Suppose 
$\cont_R = \cont_S \contComposition \cont_{S'} = ([0, E_R], [0, R])$ for some $R \in \langSet_\emptyset$. Then the
environments must have signature $(O \setunion O', I \setint I')$, and the
implementations $(I\setint I', O \setunion O')$. 

Finally, as usual, 
$
    E_R = R \setunion \missingExt(R, R, O \setunion O')$ and $
    M_R = R \setunion \missingExt(R,\allowbreak R, I \setint I')
$ are the maximal environment and implementation. $R$ is determined as follows:

\begin{proposition}\label{hbaixoib} Let $\cont_S$ and $\cont_{S'}$ be interface hypercontracts and let $\cont_R \defined \cont_S \contComposition \cont_{S'}$. Then
$R$ is given by the expression $$R = \left(S
    \setint S'\right) \setminus \left[ \uncont(S', S, O, O') \setunion
    \uncont(S, S', O', O) \right].$$
\end{proposition}

The quotient for interface hypercontracts follows from Proposition
\ref{iaxbisoudy}.

\subsection{Connection with interface automata}

Now we explore the relation of interface hypercontracts with interface
automata. Let $(I,O)$ be an IO signature. An $I$-interface automaton
\cite{alfaroHenzingerIntAutomata} is a tuple $\intAut = (Q, q_0, \to)$, where
$Q$ is a finite set whose elements we call states, $q_0 \in Q$ is the initial
state, and $\to \subseteq Q \times \alphabet \times Q$ is a deterministic
transition relation (there is at most one next state for every symbol of
$\alphabet$). We let $\intAutClass_I$ be the class of $I$-interface automata,
and $\intAutClass = \oplus_{I \in 2^\alphabet} \intAutClass_I$. In the language
of interface automata, input and output symbols are referred to as
\emph{actions}.

Given two interface automata (IA) $\intAut_i = (Q_i, q_{i,0}, \to_i) \in
\intAutClass_I$ for $i \in \{1,2\}$, we say that the state $q_1 \in Q_1$ refines
$q_2 \in Q_2$, written $q_1 \le q_2$, if
\begin{itemize}[wide=0pt]
    \item $\forall \sigma \in O, q_1' \in Q_1.\,\, q_1 \transition{\sigma}{1}
    q_1' \Rightarrow \exists q_2' \in Q_2.\,\, q_2 \transition{\sigma}{2} q_2'
    \text{ and } q_1' \le q_2'$ and
    \item $\forall \sigma \in I, q_2' \in Q_2.\,\, q_2 \transition{\sigma}{2}
    q_2' \Rightarrow \exists q_1' \in Q_1.\,\, q_1 \transition{\sigma}{1} q_1'
    \text{ and } q_1' \le q_2'$.
\end{itemize}
We say that $\intAut_1$ refines $\intAut_2$, written $\intAut_1 \le \intAut_2$,
if $q_{1,0} \le q_{2,0}$. This defines a preorder in $\intAutClass_I$.

\subsubsection{Mapping to interface hypercontracts} Suppose $\intAut = (Q, q_0,
\to) \in \intAutClass_I$. We define the language of $\intAut$, denoted
$\lang{\intAut}$, as the set of words obtained by ``playing out'' the transition
relation, i.e., $$\lang{\intAut} = \setArg{\sigma_0\sigma_1\ldots\sigma_n}{\exists
q_1, \ldots, q_{n-1}.\,\, q_i \stackrel{\sigma_i}{\to} q_{i+1} \text{ for } 0
\le i < n}.$$ Since $\lang{\intAut}$ is prefix-closed, it is an element of
$\langSet_\emptyset$.

From Section \ref{bstxiauybo}, we know that interface hypercontracts are
isomorphic to a language $S$ of $\langSet_\emptyset$ and an IO signature $I$.
The operation $\intAut \mapsto \lang{\intAut}$ maps an $I$-receptive interface automaton
$\intAut$ to a language of $\langSet_\emptyset$. Composing this map with the map (\ref{usydtfsuy})
discussed in Section \ref{bstxiauybo}, we have maps $\intAutClass \to
\langSet_\emptyset, 2^{\alphabet} \xrightarrow{\sim} \catContract$.

Thus, the interface hypercontract associated to $\intAut \in \intAutClass_I$ is
$\cont_\intAut = ( [0, \iaEnv{\intAut}], \allowbreak [0 , \lang{\intAut}] )$, where
$\iaEnv{\intAut} \in \langSet_{O}$ is given by \eqref{8b7tqotbso}. The following
result tells us that refinement of interface automata is equivalent to
refinement of their associated hypercontracts.

\begin{proposition}\label{ibgoicuign} Let $\intAut_1, \intAut_2 \in
    \intAutClass_I$. Then $\intAut_1 \le \intAut_2$ if and only if
    $\cont_{\intAut_1} \le \cont_{\intAut_2}$.
\end{proposition}

\subsubsection{Composition} 
\label{iabydanio}

Let $\intAut_1 = (Q_1, q_{1,0}, \to_1) \in \intAutClass_{I_1}$ and $\intAut_2 =
(Q_2, q_{2,0}, \to_2) \in \intAutClass_{I_2}$. The composition of the two IA is
defined if $I_1 \setunion I_2 = \alphabet$. In that case, the resulting IA,
$\intAut_1 \contComposition \intAut_2$, has IO signature $(I_1 \setint I_2, O_1
\setunion O_2)$. The elements of the composite IA are $(Q, (q_{1,0}, q_{2,0}),
\to_c)$, where the set of states and the transition relation are obtained
through the following algorithm:
\begin{itemize}[wide=0pt]
    \item Initialize $Q \coloneqq Q_1 \times Q_2$. For every $\sigma \in
    \alphabet$, $(q_1, q_2) \transition{\sigma}{c} (q_1', q_2')$ if $q_1
    \transition{\sigma}{1} q_1'$ and $q_2 \transition{\sigma}{2} q_2'$.
    \item Initialize the set of invalid states to those states where one
    interface automaton can generate an output action which the other interface
    automaton does not accept:
    \begin{align*}N \coloneqq \setArg{\begin{aligned} (q_1, q_2) \in \\ Q_1 \times Q_2 \end{aligned}}{
        \begin{aligned}
            \exists\, q_2' & \in Q_2, \sigma \in O_2\, \forall q_1' \in Q_1.\,\, \\ & q_2 \transition{\sigma}{2} q_2' \land \neg\left( q_1 \transition{\sigma}{1} q_1'\right) \text{ or}\\
            \exists\, q_1' & \in Q_1, \sigma \in O_1\, \forall q_2' \in Q_2.\,\, \\ & q_1 \transition{\sigma}{1} q_1' \land \neg\left( q_2 \transition{\sigma}{2} q_2'\right)
    \end{aligned}}.\end{align*}
    \item Also deem invalid a state such that an output action of one of the
    interface automata makes a transition to an invalid state, i.e., iterate the
    following rule until convergence:
    \begin{align*}N \coloneqq N \setunion
        \setArg{ \begin{aligned} (q_1, q_2) \in \\ Q_1 \times Q_2 \end{aligned}}{
        \begin{aligned}
            &\exists\, (q_1', q_2') \in N, \sigma \in O_1 \setunion O_2. \\ & (q_1,q_2) \transition{\sigma}{c} (q_1', q_2')
    \end{aligned}}.\end{align*}
    \item Now remove the invalid states from the IA:
    \begin{align*}Q & \coloneqq Q \setminus N \text{  and  } \\ \to_c & \coloneqq\,\, \to_c
    \setminus  \setArg{(q, \sigma, q') \in \to_c}{q \in N \text{ or } q' \in
    N}.\end{align*}
\end{itemize}
It turns out that composing IA is equivalent to composing their associated
hypercontracts:

\begin{proposition}\label{khjzgaoiui} Let $\intAut_1, \intAut_2 \in
    \intAutClass_I$. Then $\cont_{\intAut_1 \contComposition \intAut_2} =
    \cont_{\intAut_1} \contComposition \cont_{\intAut_2}$.
\end{proposition}
Propositions~\ref{ibgoicuign} and~\ref{khjzgaoiui} express that our model of
interface hypercontracts is equivalent to Interface Automata. We observe that
the definition for the parallel composition of interface hypercontracts is
straightforward, unlike for the Interface Automata (the latter involves the
iterative pruning of invalid states). In fact, in our case this pruning is
hidden behind the formula (\ref{Uncont}) defining the set Unc$()$.

\section{Conclusions}


We proposed hypercontracts, a generic model of contracts providing a richer
algebra than the metatheory of~\cite{BenvenisteContractBook}. We started from a
generic model of components equipped with a simulation preorder and  parallel
composition. On top of them, we considered \compsets (or hyperproperties, for behavioral formalisms), which
are lattices of sets of components equipped with parallel composition and
quotient; \compsets are our generic model formalizing ``properties.''
Hypercontracts are then defined as pairs of \compsets specifying the allowed
environments and either the obligations of the closed system or the set of
allowed implementations---both forms are useful. 

We specialized hypercontracts by restricting them to pairs of downward closed
\compsets (where downward closed refers to the component preorder), and then to
\conic hypercontracts, whose environments and closed systems are described by
a finite number of components. \Conic hypercontracts include Assume/Guarantee
contracts as a specialization. 
We illustrated the versatility of our model on the
definition of contracts for information flow in security.

The flexibility and power of our model suggests that a number of directions that
were opened in~\cite{BenvenisteContractBook}, but not explored to their end, can
now be re-investigated with more powerful tools: contracts and testing,
subcontract synthesis (for requirement engineering), contracts and abstract
interpretation, contracts in physical system modeling.\footnote{ 
\href{https://www.mathworks.com/products/simulink.html}{Simulink}
and \href{https://openmodelica.org/}{Modelica} toolsuites propose requirements
toolboxes, in which requirements are physical system properties that can be
tested on a given system model, thus providing a limited form of contract. This
motivates the development of a richer contract framework helping for requirement
engineering in Cyber-Physical Systems design.} Furthermore, contracts were also
developed in the neighbor community of control, which motivates us to establish
further links. In particular, Phan-Minh and Murray
\cite{phan2019contracts,phan2021contract} introduced the notion of reactive
contracts. Saoud et al.~\cite{DBLP:conf/eucc/SaoudGF18,saoud:hal-02196511}
proposed a framework of Assume/Guarantee contracts for input/output discrete or
continuous time systems. Assumptions vs. Guarantees are properties stated on
inputs vs. outputs; with this restriction, reactive contracts are considered and
an elegant formula is proposed for the parallel composition of contracts. 



\bibliographystyle{style/splncs04}
\bibliography{support/references}

\appendix

\section{Proofs}
\subsection{Proofs: Hypercontracts}

\begin{proof}[Proposition \ref{iaxbisoudy}]
    \begin{align*}
        \nonumber
        \cont'' \cmpQuotient \cont &= \bigvee \setArg{\cont'}{ \cont \contComposition \cont' \le \cont''} 
        =
        \bigvee
        \setArg{(\envSet', \impSet')}{
        \left[\begin{array}{ll}
            \impSet \cmpComposition \impSet' \le \impSet'', \\
            \envSet'' \cmpComposition \impSet  \le \envSet', \text{ and}\\ 
            \envSet'' \cmpComposition \impSet' \le \envSet
        \end{array}\right]}
        \\ &=
        \left( \left( \bigvee
        \setArg{(\envSet', \impSet')}{
        \left[\begin{array}{ll}
            \impSet \cmpComposition \impSet' \le \impSet'', \\
            \envSet'' \cmpComposition \impSet  \le \envSet', \text{ and}\\ 
            \envSet'' \cmpComposition \impSet' \le \envSet
        \end{array}\right]}
        \right) ^{-1}\right)^{-1} \\ &\stackrel{\eqref{eq:recipOptIdentity}}{=}
        \left( \bigwedge
        \setArg{(\impSet', \envSet')}{
        \left[\begin{array}{ll}
            \impSet \cmpComposition \impSet' \le \impSet'', \\
            \envSet'' \cmpComposition \impSet  \le \envSet', \text{ and}\\ 
            \envSet'' \cmpComposition \impSet' \le \envSet
        \end{array}\right]}
        \right)^{-1}
        \\ &=
        \left( \bigwedge
        \setArg{(\impSet', \envSet')}{
        \left[\begin{array}{ll}
            \envSet'' \cmpComposition \impSet  \le \envSet',\\ 
            \impSet'  \cmpComposition \impSet  \le \impSet'', \text{ and}\\
            \impSet'  \cmpComposition \envSet'' \le \envSet
        \end{array}\right]}
        \right)^{-1} \\
        &=\left( (\cont'')^{-1} \contComposition \cont \right)^{-1}.\quad\quad
    \end{align*}
\end{proof}

\begin{proof}[Proposition \ref{uiqtdvo}]
    We consider contract composition. Let $\cont = (\envSet, \impSet)$ and $\cont' = (\envSet', \impSet')$ with $\envSet = [L_e, R_e]$, $\impSet = [L_i, R_i]$, $\envSet' = [L_e', R_e']$, and $\impSet' = [L_i', R_i']$. Their composition of these two contracts $\cont \contComposition \cont' = (\envSet_c, \impSet_c)$ requires us to compute
    \begin{align*}
        &\envSet_c = \left(\envSet'\contQuotient(\impSet\contQuotient\envSet)\right)\meet\left(\envSet\contQuotient(\impSet'\contQuotient\envSet')\right) \text{  and }
        \impSet_c =
        \left((\impSet\contQuotient\envSet) \cmpComposition (\impSet'\contQuotient\envSet')\right)\contQuotient\envSet_c.
    \end{align*}
    
    Since we have to compute $\impSet\contQuotient\envSet$, we must have $L_i \le L_e$. Similarly, to compute $\impSet'\contQuotient\envSet'$, we need $L_i' \le L_e'$. Now, to compute $\envSet'\contQuotient(\impSet\contQuotient\envSet)$ and $\envSet\contQuotient(\impSet'\contQuotient\envSet')$, we must have $L_e' \le L_i$ and $L_e \le L_i'$. Then
    \begin{align*}
        L_e' \le L_i \le L_e \text{ and } L_e \le L_i' \le L_e', 
    \end{align*}
    so we must have $L_e = L_i = L_i' = L_e'$ for contract composition to be well defined. To simplify notation, let $L = L_e = L_i = L_i' = L_e'$. We obtain
    \begin{align*}
        \envSet_c &= \left(\envSet'\contQuotient(\impSet\contQuotient\envSet)\right)\meet\left(\envSet\contQuotient(\impSet'\contQuotient\envSet')\right) =
        \left(\envSet'\contQuotient([L, R_i \setunion \neg R_e])\right)\meet\left(\envSet\contQuotient([L, R_i' \setunion \neg R_e'])\right) \\ &=
        [L, R_e' \cup \neg \left( R_i \setunion \neg R_e \right) ]\meet [L, R_e \setunion \neg \left( R_i' \setunion \neg R_e' \right)] \\ &=
        [L, (R_e \setint R_e') \setunion (R_e' \setint \neg R_i') \setunion (R_e \setint \neg R_i) ]
        \intertext{and}
        \impSet_c &=
        \left((\impSet\contQuotient\envSet) \cmpComposition (\impSet'\contQuotient\envSet')\right)\contQuotient\envSet_c =
        \left([L, R_i \setunion \neg R_e] \cmpComposition [L, R_i' \setunion \neg R_e']\right)\contQuotient\envSet_c \\ &=
        \left([L, (R_i \setunion \neg R_e) \setint (R_i' \setunion \neg R_e') ]\right)\contQuotient\envSet_c =
        [L, (R_i \setunion \neg R_e) \setint (R_i' \setunion \neg R_e') ].
    \end{align*}
    
    Finally, we seek expressions for the residual. Let $\cont = (\envSet, \impSet)$ and $\cont'' = (\envSet'', \impSet'')$ with $\envSet'' = [L_e'', R_e'']$, $\impSet'' = [L_i'', R_i'']$, $\envSet = [L_e, R_e]$, and $\impSet = [L_i, R_i]$. The residual $\cont'' \contQuotient \cont = (\envSet_r, \impSet_r)$ is given by
    \begin{align*}
        &\envSet_r = \envSet'' \cmpComposition (\impSet\contQuotient\envSet) \text{  and }
        \impSet_r = \left( (\envSet \contQuotient \envSet'') \meet \left((\impSet''\contQuotient \envSet'')\contQuotient(\impSet \contQuotient \envSet)\right) \right) \contQuotient \envSet_r.
    \end{align*}
    To compute $\impSet\contQuotient\envSet$, $\envSet \contQuotient \envSet''$, and $\impSet''\contQuotient \envSet''$, we must have
    \[
        L_i \le L_e \le L_e'', \text{ and } L_i'' \le L_e''.
    \]
    We compute
    \begin{align*}
        &\envSet_r = \envSet'' \cmpComposition [L_i, R_i \setunion \neg R_e] = [L_i, R_e'' \setint \left( R_i \setunion \neg R_e \right)] \text{ and}\\
        &\impSet_r = \left( (\envSet \contQuotient \envSet'') \meet \left((\impSet''\contQuotient \envSet'')\contQuotient(\impSet \contQuotient \envSet)\right) \right) \contQuotient \envSet_r = \\ &
        \left( [L_e, R_e \setunion \neg R_e''] \meet \left([L_i'', R_i'' \setunion \neg R_e'']\contQuotient [L_i, R_i \setunion \neg R_e] \right) \right) \contQuotient \envSet_r.
    \end{align*}
    We have the further constraint $L_i'' \le L_i$. Thus, $L_i'' \le L_i \le L_e \le L_e''$ and
    {
    \begin{align*}
        \impSet_r=&
        \left( [L_e, R_e \setunion \neg R_e''] \meet  [L_i'', \left(R_i'' \setunion \neg R_e''\right) \setunion \left( R_e \setint \neg R_i \right)] \right) \contQuotient \envSet_r \\ =&
        [L_e \setunion L_i'', \left(R_e \setunion \neg R_e''\right) \setint \left(\left(R_i'' \setunion \neg R_e''\right) \setunion \left( R_e \setint \neg R_i \right)\right)] \contQuotient \envSet_r \\ =&
        [L_e , \left(R_e \setint R_i''\right) \setunion \neg R_e'' \setunion \left( R_e \setint \neg R_i \right)] \contQuotient [L_i, R_e'' \setint \left( R_i \setunion \neg R_e \right)].
        \intertext{We have the additional constrraint $L_e \le L_i$. Thus, we have $L_e = L_i = L$, and we have $L_i'' \le L \le L_e''$ and}
        \impSet_r=&
        [L , \left(R_e \setint R_i''\right) \setunion \neg R_e'' \setunion \left( R_e \setint \neg R_i \right)].
        \quad\quad
    \end{align*}
    }
\end{proof}
\subsection{Proofs: Receptive languages and hypercontracts}


\begin{proof}[Proof of Proposition \ref{kbx8we7d}]
    Suppose $L, L' \in \langSet_{I}$. If $w$ is contained in $L \setint L'$, and
    $w_p$ is a prefix of $w$, then $w$ is contained in both $L$ and $L'$, and so
    is $w_p$, which means intersection is prefix-closed.
    Moreover, for any $w' \in I^*$, we have $w \concat w' \in L$ and $w \concat
    w' \in L'$, so $w \concat w' \in L \setint L'$. We conclude that $L \setint
    L' \in \langSet_{I}$.

    Similarly, if $w$ is contained in $L \setunion L'$, then we may assume that
    $w \in L$. Any prefix $w_p$ of $w$ is also contained in $L$, so $w_p \in L
    \setunion L'$, meaning that union is prefix-closed. In addition, for every
    $w' \in I^*$, we have $w \concat w' \in L$, so $w \concat w' \in L \setunion
    L'$. This means that $L \setunion L' \in \langSet_{I}$.
\end{proof}

\begin{proof}[Proof of Proposition \ref{kjgh9yrbp}]
    First we show that $L' \rightarrow L \in \langSet_I$. Let $w \in L'
    \rightarrow L$. If $w_p$ is a prefix of $w$ then $\pre(w_p) \setint L'
    \subseteq \pre(w) \setint L' \subseteq L$, so $L' \rightarrow L$ is
    prefix-closed.

    Now suppose $w \in L' \rightarrow L$ and $w \in L$. Then for $w_I \in I^*$,
    $w \concat w_I \in L$, so $\pre(w \concat w_I) \subseteq L$. Suppose $w \in L' \rightarrow L$ and $w \not\in L$. Let $n$ be the length of
    $w$. Since $w \not\in L$, $n > 0$ (the empty string is in $L$). Write $w =
    \sigma_1 \ldots \sigma_n$ for $\sigma_i \in \alphabet$. Let $k \le n$ be the
    largest natural number such that $\sigma_1\ldots\sigma_k \in L'$ (note that
    $k$ can be zero). If $k = n$, then $w \in L' \setint \pre(w) \subseteq L$,
    which is forbidden by our assumption that $w \not\in L$. Thus, $k < n$.
    Define $w_p = \sigma_1\ldots\sigma_{k+1}$. Clearly, $w_p \not\in L'$. For
    any $w_{\alphabet} \in \alphabet^*$, since $L'$ is prefix-closed, we must
    have $\pre(w\concat w_\alphabet) \setint L' = \pre(w_p) \setint L' = \pre(w)
    \setint L' \subseteq L$. We showed that any word of $L' \rightarrow L$
    extended by a word of $I^*$ remains in $L' \rightarrow L$. We conclude that
    $L' \rightarrow L \in \langSet_I$.

    Now we show that $L' \rightarrow L$ has the properties of the exponential.
    Suppose $L'' {\in} \langSet_I$ is such that $L' {\setint} L'' \subseteq L$. Let
    $w {\in} L''$. Then $\pre(w) \setint L' \subseteq L$, which means that $L''
    \le L' \rightarrow L$. On the other hand, 
    \[
        L' \setint \left( L' \rightarrow L \right) = L' \setint \setArg{w \in \alphabet^*}{\pre(w) \setint L' \subseteq L}
        \subseteq L.
    \]
    Thus, any $L'' \le L' \rightarrow L$ satisfies $L'' \setint L' \le L$. This
    concludes the proof.
\end{proof}

\begin{proof}[Proof of Proposition \ref{98wninhiu}]
From Prop.~\ref{kjgh9yrbp}, it is clear that $L \subseteq L' \rightarrow L$.
Suppose $w \in L' \setint L$. Since $L$ and $L'$ are $I$-receptive, $w\circ
\sigma \in L \setint L'$ for $\sigma \in I$. Assume $\sigma \in O$. If $w \circ
\sigma \not \in L'$, then we can extend $w \circ \sigma$ by any word $w' \in
\alphabet^*$, and this will satisfy $\pre(w\circ \sigma \circ w') \setint L' =
\pre(w) \setint L' \subseteq L$ due to the fact $L'$ is prefix-closed. If $w
\concat \sigma \in L' \setminus L$, then $w \not \in L' \rightarrow L$. Thus, we
can express the exponential using the closed-form expression of the proposition.
\end{proof}


\begin{proof}[Proof of Proposition \ref{nha9dbq9o}]
    Suppose $w \in L / L'$ and $w \in L \setint L'$. We have not lost generality
    because $\estring \in L \setint L'$. We consider extensions of $w$ by a
    symbol $\sigma$:
    \begin{enumerate}[label=\alph*.]
        \item If $\sigma \in I$, $\sigma$ is an input symbol for both $L'$ and
        the quotient.
        \begin{enumerate}[label=\roman*.]
            \item $L$ is receptive to $I$, so $w \concat \sigma \in L$;
            \item $L'$ is receptive to $I \subseteq I'$, so $w \concat \sigma
            \in L'$; and
            \item $L / L'$ must contain $w \concat \sigma$ because the quotient
            is $I_r$-receptive.
        \end{enumerate}
        
        \item If $\sigma \in O \setint I'$, then $\sigma$ is an output of the
        quotient, and an input of $L'$.
        \begin{enumerate}[label=\roman*.]
            \item $L'$ is $I'$-receptive, so $w \concat \sigma \in L'$;
            \item $\sigma$ is an output symbol for both $L$ and $L / L'$, so
            none of them is required to contain $w \concat \sigma$; and
            \item if $w \concat \sigma \in L' \setminus L$, the extension $w
            \concat \sigma$ cannot be in the quotient. Otherwise, it can.
        \end{enumerate}
        
        \item If $\sigma \in O'$, $\sigma$ is an output for $L'$ and an input
        for the quotient.
        \begin{enumerate}[label=\roman*.]
            \item Neither $L$ nor $L'$ are $O'$-receptive;
            \item $L / L'$ is $O'$-receptive, so we must have $w \concat \sigma
            \in L / L'$; and
            \item if $w \concat \sigma \in L' \setminus L$, we cannot have
            $w \concat \sigma \in L / L'$.
        \end{enumerate}
    \end{enumerate}
    Starting with a word $w$ in the quotient, statements a and b allow or
    disallow extensions of that word to be in the quotient. However, statements
    c.ii and c.iii impose a requirement on the word $w$ itself, i.e., if c.iii
    is violated, c.ii implies that $w$ is not in the quotient. Statements a.iii
    and c.ii impose an obligation on the quotient to accept  extensions by
    symbols of $I$ and $O'$; and those extensions may lead to a violation of
    c.iii. Thus, we remove from the quotient all words such that extensions of
    those words by elements of $I \setunion O'$ end up in $L' \setminus L$. The
    expression of the proposition follows from these considerations.
\end{proof}

\begin{proof}[Proof of Proposition \ref{hbaixoib}]
    From the principle of hypercontract composition, we must have
    \begin{align}
        \label{djkabhpowby}
        &E_R \le U \defined (E_{S'} / M_S ) \land (E_S / M_{S'})  \text{ and} \\
        \label{kbaodiaynoi}
        &L \defined M_{S'} \times M_S \le M_R.
    \end{align}
    Observe that the quotients $E_{S'} / M_S$ and $E_S / M_{S'}$ both have IO
    signature $O \setunion O'$, so the conjunction in \eqref{djkabhpowby} is
    well-defined as an operation of the Heyting algebra $\langSet_{O \setunion
    O'}$. We study the first element:
    \begin{align*}
        E_{S'} / M_S = & \left(E_{S'} \setint M_{S} \setunion  \missingExt(E_{S'}, M_{S}, O) \right) \setminus \\ 
        & \uncont(E_{S'}, M_{S}, O, O').
    \end{align*}
    We attempt to simplify the terms. Suppose $w \in E_{S'} \setint (M_{S}
    \setminus S)$. Then all extensions of $w$ lie in $M_{S} \setminus S$. This
    means that $\missingExt(E_{S'}, M_{S}, O) = \missingExt(E_{S'}, S,\allowbreak O)$.
    Moreover, if a word is an element of $E_{S'} \setminus S'$, all its
    extensions are in this set, as well (i.e., it is impossible to escape this
    set by extending words). Thus, $\uncont(E_{S'}, M_{S}, O, O') = \uncont(S',
    M_{S}, O, O')$. We have
    \begin{align*}
        E_{S'} / M_S = & \left(E_{S'} \setint M_{S} \setunion  \missingExt(E_{S'}, S, O) \right) \setminus \\ 
        & \uncont(S', M_{S}, O, O').
    \end{align*}
    Now we can write
    \begin{align*}
        U = &
        \left[
            \begin{aligned}
            & \left(E_{S'} \setint M_{S} \setunion  \missingExt(E_{S'}, S, O) \right) \setint \\
            & \left(E_{S} \setint M_{S'} \setunion  \missingExt(E_{S}, S', O') \right)
            \end{aligned}
        \right] \setminus
        \\ 
        & \left[ \uncont(S', M_{S}, O, O') \setunion \uncont(S, M_{S'}, O', O) \right].
    \end{align*}
    Observe that
    {\small
    \begin{align*}
        &E_{S'} \setint M_{S} \setint \missingExt(E_{S}, S', O') \\ &= 
        (S' \setunion \missingExt(S', S', O')) \setint M_{S} \setint \missingExt(E_{S}, S', O') \\ &=
        M_{S} \setint \missingExt(E_{S}, S', O') = M_{S} \setint \missingExt(S, S', O').
    \end{align*}
    }
    The last equality comes from the following fact: if a word of
    $\missingExt(E_{S}, S', O')$ is obtained by extending a word of $(E_S
    \setminus S) \setint S'$ by $O'$, the resulting word is still an element of
    $E_S$, which means it cannot be an element of $M_S$ because $M_S$ and $E_S$
    are disjoint outside of $S$.
    Therefore,
    {\small
    \begin{equation}\label{xiqndp08n}
        U =
            \begin{aligned} &
            \left[
            \begin{aligned}
            &\left(S \setint S'\right) \setunion \\
            &\left(M_{S} \setint \missingExt(S, S', O') \right) \setunion \\
            &\left(M_{S'} \setint \missingExt(S', S, O) \right) \setunion \\
            &\left(\missingExt(E_{S'}, S, O) \setint \missingExt(E_S, S', O') \right)
            \end{aligned}
            \right]
            \setminus \\ & \left[ \uncont(S', M_{S}, O, O') \setunion \uncont(S, M_{S'}, O', O) \right].
            \end{aligned}
    \end{equation} }

    We can write
    \begin{align*}
        \missingExt(&E_{S'}, S, O) \setint \missingExt(E_S, S', O') = \\
        &\missingExt(E_{S'}, S, O) \setint \missingExt(S, S', O') \;\setunion \\
        &\missingExt(S', S, O) \setint \missingExt(E_S, S', O') \;\setunion \\
        &\left( \begin{aligned} & \missingExt(\missingExt(S', S', O'), S, O) \setint \\ &\missingExt(\missingExt(S, S, O), S', O') \end{aligned}\right).
    \end{align*}
    Note that $\missingExt(E_{S'}, S, O) \setint \missingExt(S, S', O') = \allowbreak
    \missingExt(S, S, O) \setint \missingExt(\allowbreak 
    S, S', O')$. Hence
    \begin{align*}
        &\left(M_{S} \setint \missingExt(S, S', O') \right) \setunion \\ &\left(\missingExt(E_{S'}, S, O) \setint \missingExt(S, S', O') \right) \\
        &=
        \missingExt(S, S', O') \setint \left( M_{S} \setunion  \missingExt(E_{S'}, S, O) \right) \\ &=
        \missingExt(S, S', O') \setint \left( M_{S} \setunion  \missingExt(S, S, O) \right) \\ &=
        \missingExt(S, S', O').
    \end{align*}
    Finally, we observe that the set $\missingExt(\missingExt(S', S', O'), S, O)
    \setint \missingExt(\allowbreak \missingExt(S, S, O), S', O')$ must be empty since the
    words of the first term have prefixes in $S \setminus S'$, and the second in
    $S' \setminus S$. These considerations allow us to conclude that
    \begin{align*}
        U \le
            \begin{aligned} &
            \left[
            \begin{aligned}
            &\left(S \setint S'\right) \;\setunion \\
            &\missingExt(S, S', O')   \;\setunion \\
            &\missingExt(S', S, O).
            \end{aligned}
            \right]
            \setminus \\ & \bigl[ \uncont(S', M_{S}, O, O') \setunion \uncont(S, M_{S'}, O', O) \bigr].
            \end{aligned}
    \end{align*}
    To simplify the expression a step further, suppose $w \in \uncont(S', M_{S},
    O, O')$ and has a prefix in $S' \setint (M_S \setminus S)$. Then $w \not \in
    S \setint S'$. The words of $\missingExt(S, S', \allowbreak O')$ do not have prefixes in
    $S' \setminus S$, so $w \not \in \missingExt(S, S', O')$. The words of
    $\missingExt(S', S, O)$ belong to $E_S$, which is disjoint from $M_S$
    outside of $S$. Thus, $w \not \in \missingExt(S', S, O)$.
    
    We just learned that the words of $\uncont(S', M_{S}, O, O')$ having a
    prefix in $S' \setint (M_S \setminus S)$ are irrelevant for the inequality
    above. Now consider a word $w$ of $\uncont(S', M_{S}, O, O')$ with no prefix
    in $S' \setint (M_S \setminus S)$. Let $w_p$ be the longest prefix of $w$
    which is in $S \setint S'$. There is a word $w' \in (O \setunion O')^*$ and
    a symbol $\sigma \in O$ such that $w_p \concat w' \in S' \setint M_S$ and
    $w_p \concat w' \concat \sigma \in M_S \setminus S'$. Suppose $w'$ is not
    the empty string. Then we can let $\sigma'$ be the first symbol of $w'$.
    Then $w_p \concat \sigma' \in M_S \setminus S$, so $\sigma' \in O'$. But
    this means that $w \in \uncont(S, S', O', O)$. If $w'$ is empty, $w_p \in S
    \setint S'$ and $w_p \concat \sigma' \in M_S \setminus S'$. Since $\sigma
    \in O$, $w_p \concat \sigma \in \setint M_S$ if and only if it belongs to
    $S$. Thus, $w_p \concat \sigma' \in S \setminus S'$, which means that $w \in
    \uncont(S', S, O, O')$. We can thus simplify the upper bound on $E_R$ to
    \begin{align}\label{lkjnha}
        U =
            \begin{aligned} &
            \left[
            \begin{aligned}
            &\left(S \setint S'\right) \; \setunion \\
            &\missingExt(S, S', O') \;  \setunion \\
            &\missingExt(S', S, O)
            \end{aligned}
            \right]
            \setminus \\ & \left[ \uncont(S', S, O, O') \setunion \uncont(S, S', O', O) \right].
            \end{aligned}
    \end{align}
    
    Define $\hat R \defined \left(S \setint S'\right) \setminus \left[ \uncont(S', S, O, O')
    \setunion \uncont(S, S', O', O) \right]$. We want to show that $U{=}$ \linebreak $\hat R \setunion \missingExt(\hat R, \hat R, O \setunion O')$. Note that we only have to prove that
    \begin{align}
        \nonumber
        \missingExt&(\hat R, \hat R, O \setunion O') = \\
        \label{bhstdoit}
        & \begin{aligned} &
    \left[
            \begin{aligned}
            &\missingExt(S, S', O') \;  \setunion \\
            &\missingExt(S', S, O)
            \end{aligned}
            \right]
            \setminus \\ & \left[ \uncont(S', S, O, O') \setunion \uncont(S, S', O', O) \right].
        \end{aligned}
    \end{align}

    \begin{proof}[Proof of \eqref{bhstdoit}]
    Suppose $w \in \missingExt(S, S', O') \setminus [ \uncont(S', S, O, O')
    \setunion \allowbreak \uncont(S, S', O', O) ]$. Write $w = w_p \concat \sigma
    \concat w'$, where $w_p$ is the longest prefix of $w$ which lies in $S
    \setint S'$, $\sigma \in O'$, and $w' \in \alphabet^*$. $w_p \not \in
    \left[ \uncont(S', S, O, O') \setunion \uncont(S, S', O', O) \right]$
    because all its extensions would be in this set if $w_p$ were in this set,
    and we know that $w$ is not in this set. It follows that $w_p \in \hat R$
    and since $w_p \concat \sigma \not \in \hat R$, $w_p \concat \sigma$ and all
    its extensions are in $\hat R \setunion \missingExt(\hat R,\hat R, O
    \setunion O')$. Thus, $w \in \hat R \setunion \missingExt(\hat R,\hat R, O
    \setunion O')$
    
    The same argument applies when $w \in \missingExt(S', S, O) \setminus
    [ \uncont(S', S, O, O') \allowbreak \setunion \uncont(S, S', O', O) ]$.
    We conclude that the right hand side of \eqref{bhstdoit} is a subset of
    the left hand side.
    
    Now suppose that $w \in \missingExt(\hat R,\hat R, O \setunion O')$ and
    write $w = w_p \concat \sigma \concat w'$, where $w_p$ is the longest prefix
    of $w$ contained in $\hat R$, $\sigma \in O \setunion O'$, and $w' \in
    \alphabet^*$.
    From the definition of $\hat R$, $w_p \in S \setint S'$. Suppose $w_p
    \concat \sigma \in S \setint S'$. Then
    \[w_p \concat \sigma \in \left[ \uncont(S', S, O, O') \setunion \uncont(S,
    S', O', O) \right],\] which means that $w_p$ also belongs to this set
    (because $\sigma \in O \setunion O'$). This contradicts the fact that $w_p
    \in \hat R$, so our assumption that $w_p \concat \sigma \in S \setint S'$ is
    wrong. Then $w_p$ is also the longest prefix of $w$ contained in $S \setint
    S'$.
    
    Without loss of generality, assume $\sigma \in O$. Suppose $w_p \concat
    \sigma \not \in S$. Then $w \in \missingExt(S', S, O)$. Moreover, since $w_p
    \in \hat R$, $w_p \concat \sigma \not \in [ \uncont(S', S, O, O')
    \setunion \allowbreak \uncont(S, S', O', O) ]$. Since $w_p \concat \sigma \not \in
    S \setint S'$, we have $w \not \in [ \uncont(S', S, O, O') \setunion \allowbreak
    \uncont(S, S', O', O) ]$. Thus, $w$ is in the right hand set of
    \eqref{bhstdoit}.
    
    Now suppose $w_p \concat \sigma \not \in S'$. If $w_p \concat \sigma \in S$,
    then $w_p \in \uncont(S', S, O, O')$, which contradicts the fact that $w_p \in
    R$. We must have $w_p \concat \sigma \not \in S$, which we already showed
    implies that $w$ is in the right hand set of \eqref{bhstdoit}.
    
    An analogous reasoning applies to $\sigma \in O'$. We conclude that the
    right hand side of \eqref{bhstdoit} is a subset of the left hand side, and
    this finishes the proof of their equality.
    \end{proof}

    This result and \eqref{djkabhpowby} tell us that $E_R \le \hat R \setunion
    \missingExt(\hat R, \hat R, O \setunion O')$. Now we study the constraint
    \eqref{kbaodiaynoi}. We want to show that $\hat R$ yields the tightest bound
    $L \le \hat R \setunion \missingExt(\hat R, \hat R, I \setint I')$ which
    also respects the bound \eqref{djkabhpowby}.

    \begin{proof}
        Observe that $L = 
        \left( S' \setunion \missingExt(S', S', I') \right) \setint 
        \left( S \setunion \missingExt(S, S, I) \right)$.
        First we will show that $L \subseteq \hat R \setunion \missingExt(\hat
        R, \hat R, I \setint I')$. Suppose $w \in L$. Then $w$ belongs to at
        least one of the sets $(1)\; S \setint S'$, $(2)\; S \setint \missingExt(S', S',
        I')$, $(3)\; S' \setint \missingExt(S, S, I)$, or  $(4)\; \missingExt(S, S, I)
        \setint \missingExt(S', S', I')$. We analyze each case:
        \begin{enumerate}
            \item Suppose $w \in S \setint S'$. If $w \in \hat R$, then clearly
            $w \in \hat R \setunion \missingExt(\hat R, \hat R, I \setint I')$.
            Suppose $w \not \in \hat R$. Then there is word $w' \in (O \setunion
            O')^*$ and either a symbol $\sigma \in O$ such that $w \concat w'
            \concat \sigma \in S \setminus S'$ or a symbol $\sigma \in O'$ such
            that $w \concat w' \concat \sigma \in S' \setminus S$. Write $w =
            w_p \concat w''$ such that $w''$ is the longest suffix of $w$ which
            belongs to $O \setunion O'$. It follows that the last symbol of
            $w_p$ is an element of $I \setint I'$. Since $w \not \in \hat R$,
            neither does $w_p$. This shows that for every word $w_r \concat
            \sigma_r \in S\setint S'$ such that $w_r \in \hat R$ but $w_r
            \concat \sigma_r \not \in \hat R$, we must have $\sigma_r \in I
            \setint I'$.
    
            Let $w_p'$ be the longest prefix of $w_p$ which lies in $\hat R$. By
            assumption, $\hat R$ is not empty. If we write $w = w_p' \concat
            w''$, the first symbol of $w''$ is in $I \setint I'$. Thus, $w \in
            \missingExt(\hat R, \hat R, I \setint I')$.
            
            \item Observe that $\missingExt(S', S', I') = \missingExt(S', S', I'
            \setint I) \setunion \missingExt(S', S', I' \setint O)$.
            Moreover, $S \setint \missingExt(S', S', I \setint I') \subseteq
            \missingExt(S\setint S', S\setint S', I \setint I') \subseteq
            \missingExt(\hat R, \hat R, I \setint I')$.
    
            Suppose $w \in S \setint \missingExt(S', S', I' \setint O)$. Then $w
            \in \uncont(S', S, O, O')$, so $w \not \in \hat R$. Let $w_i$ be the
            longest prefix of $w$ which lies in $S \setint S'$. Then $w_i \not
            \in \hat R$, either. Let $w_p$ be the longest prefix of $w_i$ which
            is in $\hat R$. Then $w_i \in
            \missingExt(\hat R, \hat R, I \setint I')$, and therefore, so does
            $w$.
    
            \item If $w \in S' \setint \missingExt(S, S, I)$, an analogous reasoning
            applies.
            
            \item Suppose $w \in \missingExt(S, S, I) \setint \missingExt(S',
            S', I')$. If $w$ has a prefix in $S' \setint \missingExt(S, S, I)$
            or $S \setint \missingExt(S', S', I')$, then the reasoning of the
            last two points applies, and we have $w \in \missingExt(\hat R, \hat
            R, I \setint I')$.
            Suppose $w$ has no such a prefix, and write $w = w_p \concat w'$,
            where $w_p$ is the longest prefix of $w$ which lies in $S \setint
            S'$. Let $\sigma$ be the first symbol of $w'$. Then $w_p \concat
            \sigma \in \missingExt(S, S, I) \setint \missingExt(S', S', I')$,
            which means that $\sigma \in I \setint I'$. Thus, $w \in
            \missingExt(S\setint S', S\setint S', I \setint I') \subseteq
            \missingExt(\hat R, \hat R, I \setint I')$.
        \end{enumerate}
        We have shown that $L \subseteq \hat R \setunion \missingExt(\hat R,
        \hat R, I \setint I')$. Now suppose $w \in \hat R \setunion
        \missingExt(\hat R, \hat R, I \setint I')$. If $w \in \hat R$ then
        clearly $w \in S \setint S' \subseteq L$. Suppose $w \in
        \missingExt(\hat R, \hat R, I \setint I')$ and let $w_r$ be the longest
        prefix of $w$ contained in $\hat R$ and $\sigma_r$ the symbol that comes
        immediately after $w_r$ in $w$. Clearly $\sigma_r \in I \setint I'$.
    
        If $w_r \concat \sigma_r \in S \setminus S'$, then $w_r \concat
        \sigma_r$ cannot be an element of $\hat R$. If it were, we would have
        $E_{\hat R} \times S \not \subseteq E_{S'}$, violating the bound
        \eqref{djkabhpowby}. The same applies when $w_r \concat \sigma_r \in S'
        \setminus S$.
        
        If $w_r \concat \sigma_r \not \in S \setunion S'$, then $w \concat
        \sigma_r \in \missingExt(S', S', I') \setint \missingExt(S, S, I) \subseteq L$.

        If $w_r \concat \sigma_r \in S \setint S'$, then $w_r \concat \sigma_r
        \in \uncont(S', S, O, O') \setunion \uncont(S, S', O', O)$, which means
        that $w_r \concat \sigma_r$ is not allowed to be an element of $\hat R$;
        otherwise, there would be a contradiction of \eqref{djkabhpowby}.
    \end{proof}
We conclude that $R = \hat R$.
\end{proof}

\begin{proof}[Proof of Proposition \ref{ibgoicuign}]
    Suppose that $\intAut_1 \le \intAut_2$. We want to show that
    $\iaImp{\intAut_1} \le \iaImp{\intAut_2}$ and $\iaEnv{\intAut_2} \le
    \iaEnv{\intAut_1}$. We proceed by induction in the length $n$ of words,
    i.e., we will show that this relations hold for words of arbitrary length.
    
    Consider the case $n = 1$. Suppose $\sigma \in \iaImp{\intAut_1} \setint
    \alphabet$. If $\sigma \in I$, then $\sigma \in \iaImp{\intAut_2}$ because
    of $I$-receptivity. If $\sigma\in O$, then $\sigma \in \lang{\intAut}$, so
    there exists $q_1 \in Q_1$ such that $q_{1,0} \transition{\sigma}{1} q_1$,
    which means that there exists $q_2 \in Q_2$ such that $q_{2,0}
    \transition{\sigma}{2} q_2$. Thus, $\sigma \in \lang{\intAut_2} \subseteq
    \iaImp{\intAut_2}$. We have shown that $\iaImp{\intAut_1} \subseteq
    \iaImp{\intAut_2}$ for $n = 1$. An analogous reasoning shows that
    $\iaEnv{\intAut_2} \subseteq \iaEnv{\intAut_1}$

    Suppose the statement is true for words of length $n$. Let $w
    \concat \sigma \in \iaImp{\intAut_1}$, where $w \in \alphabet^*$ is a word
    of length $n$, and $\sigma \in \alphabet$. By the inductive assumption, $w
    \in \iaImp{\intAut_2}$. 
    \begin{itemize}[wide=0pt]
        \item 
    If $\sigma \in I$, then $w \concat \sigma \in
    \iaImp{\intAut_2}$ due to $I$-receptiveness.
    \item
    Let $\sigma \in O$ and $w \not \in \lang{\intAut_1}$. Then we can write $w =
    w_p \concat w'$, where $w_p$ is the longest prefix of $w$ which lies in
    $\lang{\intAut_1}$ (suppose it has length $l$). Let $\sigma'$ be the first
    symbol of $w'$; clearly $\sigma' \in I$. Since $w \in \lang{\intAut_2}$
    and this set is prefix-closed, $w_p \in \lang{\intAut_2}$. Since $w_p \in
    \lang{\intAut_1} \setint \lang{\intAut_2}$, there exist $\{q_{j,i} \in
    Q_j\}_{i = 1}^k$ (for $j \in \{1, 2\}$) such that $q_{j, i-1}
    \transition{w_i}{j} q_{j, i}$ for $0 < i \le k$, where $w_i$ is the $i$-th
    symbol of $w_p$. Since the IA are deterministic, we must have $q_{1,i} \le
    q_{2, i}$. Suppose there were a $q_2 \in Q_2$ such that $q_{2, k}
    \transition{\sigma'}{2} q_2$; since $q_{1,k} \le q_{2,k}$ and
    $\sigma' \in I$, this would mean that there exists $q_1 \in Q_1$ such that
    $q_{1, k} \transition{\sigma'}{1} q_1$, which would mean that $w_p \concat
    \sigma' \in \lang{\intAut_1}$, a contradiction. We conclude that such $q_2$
    does not exist, which means that $w_p \concat \sigma' \not \in
    \lang{\intAut_2}$, which means that $w \concat \sigma \in \iaImp{\intAut_2}$
    because of I-receptiveness.
    \item
    Finally, if $\sigma \in O$ and $w \in \lang{\intAut_1}$, then there exist
    $\{q_{1,i} \in Q_1\}_{i = 1}^n$ such that $q_{1,i-1} \transition{w_i}{1}
    q_{1, i}$ for $0 < i \le n$, where $w_i$ is the $i$-th symbol of $w$. Since
    $w \concat \sigma \in \iaImp{\intAut_1}$, $w \in \lang{\intAut_1}$, and
    $\sigma \in O$, we must have $w \concat \sigma \in \lang{\intAut_1}$. This
    means that there must exist $q_{1,n+1} \in Q_1$ such that $q_{1,n}
    \transition{\sigma}{1} q_{1,n+1}$. We know that $w \in \iaImp{\intAut_2}$ by
    the induction assumption. If $w \not \in \lang{\intAut_2}$, then clearly $w
    \concat \sigma \in \iaImp{\intAut_2}$. If $w \in \lang{\intAut_2}$, there
    are states $\{q_{2,i} \in Q_2\}_{i = 1}^n$ such that $q_{2,i-1}
    \transition{w_i}{2} q_{2, i}$ for $0 < i \le n$. Moreover, there exists
    $q_{n+1} \in Q_1$ such that $q_{n} \transition{\sigma}{1} q_{n+1}$ and
    $q_{1,n} \le q_{2,n}$, there must be a $q_{2,n+1} \in Q_2$ such that
    $q_{2,n} \transition{\sigma}{2} q_{2,n+1}$, which means that $w \concat
    \sigma \in \iaImp{\intAut_2}$.
    \end{itemize}
    We have shown that $\iaImp{\intAut_1} \subseteq \iaImp{\intAut_2}$. An
    analogous argument proves that $\iaEnv{\intAut_2} \subseteq
    \iaEnv{\intAut_1}$.
    
    Now suppose that $\iaImp{\intAut_1} \subseteq \iaImp{\intAut_2}$ and
    $\iaEnv{\intAut_2} \subseteq \iaEnv{\intAut_1}$. We want to show that
    $q_{1,0} \le q_{2,0}$. We proceed by coinduction.


    Let $n$ be a natural number. Suppose there exist sets $\{q_{j,i} \in
    Q_j\}_{i = 1}^n$ with $j \in \{1,2\}$ such that $q_{1,i} \le q_{2,i}$ for
    all $i$ and a word $w$ of length $n$ such that $q_{j,i-1}
    \transition{w_i}{j} q_{j,i}$ for $0 < i \le n$. Suppose there exists
    $q_{1,n+1} \in Q_1$ and $\sigma \in O$ such that $q_{1, n}
    \transition{\sigma}{1} q_{1,n+1}$. Then $w \concat \sigma \in
    \iaImp{\intAut_1} \subseteq \iaImp{\intAut_2}$. Observe that $w \in
    \lang{\intAut_2}$, so we must have $w \concat \sigma \in \lang{\intAut_2}$.
    This means there must be a $q_{2, n+1} \in Q_2$ such that $q_{2, n}
    \transition{\sigma}{2} q_{2, n+1}$. We assume that $q_{1, n+1} \le q_{2,
    n+1}$. Similarly, suppose there exists $q_{2, n+1}' \in Q_2$ and $\sigma \in
    I$ such that $q_{2, n} \transition{\sigma}{2} q_{2, n+1}'$. Then $w \concat
    \sigma \in \iaEnv{\intAut_2} \subseteq \iaEnv{\intAut_1}$. Since $w \in
    \lang{\intAut_1}$, we must have $w \concat \sigma \in \lang{\intAut_1}$.
    Thus, there must exist $q_{1, n+1}' \in Q_1$ such that $q_{1,n}
    \transition{\sigma}{1} q_{1,n+1}'$. We assume that $q_{1,n+1} \le
    q_{2,n+1}$.
    This finished the coinductive proof.
\end{proof}

\begin{proof}[Proof of Proposition \ref{khjzgaoiui}]
    Let $\intAut_i$ have IO signatures $(I_i, O_i)$ for $i \in \{1, 2\}$. For
    composition to be defined, we need $I_1 \setunion I_2 = \alphabet$. Let
    $\cont_{\intAut_i}$ be the interface contract associated with $\intAut_i$. From
    Proposition \ref{hbaixoib} and Section \ref{bstxiauybo}, the composition
    $\cont_{\intAut_1} \contComposition \cont_{\intAut_2}$ is isomorphic to $I_1 \setint
    I_2$ and the $\langSet_\emptyset$ language $R = (\lang{\intAut_1} \setint
    \lang{\intAut_2}) \setminus \left[ \uncont(\lang{\intAut_1},
    \lang{\intAut_2}, O_2, O_1) \setunion \uncont(\lang{\intAut_2}, \lang{\intAut_1},
    O_1, O_2) \right]$. From Section \ref{iabydanio}, we deduce that
    $\lang{\intAut_1 \contComposition \intAut_2} = R$. The proposition follows.
\end{proof}

\end{document}